\documentclass{article}
\usepackage{amsmath}
\usepackage{amssymb}\usepackage{amsthm}
\usepackage{amsfonts}
\begin{document}
\theoremstyle{plain}
\newtheorem*{ithm}{Theorem}
\newtheorem*{idefn}{Definition}
\newtheorem{thm}{Theorem}[section]
\newtheorem{lem}[thm]{Lemma}
\newtheorem{prop}[thm]{Proposition}
\newtheorem{cor}[thm]{Corollary}
\newtheorem*{icor}{Corollary}
\theoremstyle{definition}
\newtheorem{assum}[thm]{Assumption}
\newtheorem{notation}[thm]{Notation}
\newtheorem{defn}[thm]{Definition}
\newtheorem{clm}[thm]{Claim}
\newtheorem{ex}[thm]{Example}
\theoremstyle{remark}
\newtheorem{rem}[thm]{Remark}
\newcommand{\unit}{\mathbb I}
\newcommand{\ali}[1]{{\mathfrak A}_{[ #1 ,\infty)}}
\newcommand{\alm}[1]{{\mathfrak A}_{(-\infty, #1 ]}}
\newcommand{\nn}[1]{\lV #1 \rV}
\newcommand{\br}{{\mathbb R}}
\newcommand{\dm}{{\rm dom}\mu}
\newcommand{\lb}{l_{\bb}(n,n_0,k_R,k_L,\lal,\bbD,\bbG,Y)}
\newcommand{\Ad}{\mathop{\mathrm{Ad}}\nolimits}
\newcommand{\Proj}{\mathop{\mathrm{Proj}}\nolimits}
\newcommand{\RRe}{\mathop{\mathrm{Re}}\nolimits}
\newcommand{\RIm}{\mathop{\mathrm{Im}}\nolimits}
\newcommand{\Wo}{\mathop{\mathrm{Wo}}\nolimits}
\newcommand{\Prim}{\mathop{\mathrm{Prim}_1}\nolimits}
\newcommand{\Primz}{\mathop{\mathrm{Prim}}\nolimits}
\newcommand{\ClassA}{\mathop{\mathrm{ClassA}}\nolimits}
\newcommand{\Class}{\mathop{\mathrm{Class}}\nolimits}
\newcommand{\diam}{\mathop{\mathrm{diam}}\nolimits}
\def\qed{{\unskip\nobreak\hfil\penalty50
\hskip2em\hbox{}\nobreak\hfil$\square$
\parfillskip=0pt \finalhyphendemerits=0\par}\medskip}
\def\proof{\trivlist \item[\hskip \labelsep{\bf Proof.\ }]}
\def\endproof{\null\hfill\qed\endtrivlist\noindent}
\def\proofof[#1]{\trivlist \item[\hskip \labelsep{\bf Proof of #1.\ }]}
\def\endproofof{\null\hfill\qed\endtrivlist\noindent}
\newcommand{\oo}{{\boldsymbol\omega}}
\newcommand{\caA}{{\mathcal A}}
\newcommand{\caB}{{\mathcal B}}
\newcommand{\caC}{{\mathcal C}}
\newcommand{\caD}{{\mathcal D}}
\newcommand{\caE}{{\mathcal E}}
\newcommand{\caF}{{\mathcal F}}
\newcommand{\caG}{{\mathcal G}}
\newcommand{\caH}{{\mathcal H}}
\newcommand{\caI}{{\mathcal I}}
\newcommand{\caJ}{{\mathcal J}}
\newcommand{\caK}{{\mathcal K}}
\newcommand{\caL}{{\mathcal L}}
\newcommand{\caM}{{\mathcal M}}
\newcommand{\caN}{{\mathcal N}}
\newcommand{\caO}{{\mathcal O}}
\newcommand{\caP}{{\mathcal P}}
\newcommand{\caQ}{{\mathcal Q}}
\newcommand{\caR}{{\mathcal R}}
\newcommand{\caS}{{\mathcal S}}
\newcommand{\caT}{{\mathcal T}}
\newcommand{\caU}{{\mathcal U}}
\newcommand{\caV}{{\mathcal V}}
\newcommand{\caW}{{\mathcal W}}
\newcommand{\caX}{{\mathcal X}}
\newcommand{\caY}{{\mathcal Y}}
\newcommand{\caZ}{{\mathcal Z}}
\newcommand{\bbA}{{\mathbb A}}
\newcommand{\bbB}{{\mathbb B}}
\newcommand{\bbC}{{\mathbb C}}
\newcommand{\bbD}{{\mathbb D}}
\newcommand{\bbE}{{\mathbb E}}
\newcommand{\bbF}{{\mathbb F}}
\newcommand{\bbG}{{\mathbb G}}
\newcommand{\bbH}{{\mathbb H}}
\newcommand{\bbI}{{\mathbb I}}
\newcommand{\bbJ}{{\mathbb J}}
\newcommand{\bbK}{{\mathbb K}}
\newcommand{\bbL}{{\mathbb L}}
\newcommand{\bbM}{{\mathbb M}}
\newcommand{\bbN}{{\mathbb N}}
\newcommand{\bbO}{{\mathbb O}}
\newcommand{\bbP}{{\mathbb P}}
\newcommand{\bbQ}{{\mathbb Q}}
\newcommand{\bbR}{{\mathbb R}}
\newcommand{\bbS}{{\mathbb S}}
\newcommand{\bbT}{{\mathbb T}}
\newcommand{\bbU}{{\mathbb U}}
\newcommand{\bbV}{{\mathbb V}}
\newcommand{\bbW}{{\mathbb W}}
\newcommand{\bbX}{{\mathbb X}}
\newcommand{\bbY}{{\mathbb Y}}
\newcommand{\bbZ}{{\mathbb Z}}
\newcommand{\str}{^*}
\newcommand{\lv}{\left \vert}
\newcommand{\rv}{\right \vert}
\newcommand{\lV}{\left \Vert}
\newcommand{\rV}{\right \Vert}
\newcommand{\la}{\left \langle}
\newcommand{\ra}{\right \rangle}
\newcommand{\ltm}{\left \{}
\newcommand{\rtm}{\right \}}
\newcommand{\lcm}{\left [}
\newcommand{\rcm}{\right ]}
\newcommand{\ket}[1]{\lv #1 \ra}
\newcommand{\bra}[1]{\la #1 \rv}
\newcommand{\lmk}{\left (}
\newcommand{\rmk}{\right )}
\newcommand{\al}{{\mathcal A}}
\newcommand{\md}{M_d({\mathbb C})}
\newcommand{\id}{\mathop{\mathrm{id}}\nolimits}
\newcommand{\Tr}{\mathop{\mathrm{Tr}}\nolimits}
\newcommand{\Ran}{\mathop{\mathrm{Ran}}\nolimits}
\newcommand{\Ker}{\mathop{\mathrm{Ker}}\nolimits}
\newcommand{\Aut}{\mathop{\mathrm{Aut}}\nolimits}
\newcommand{\spn}{\mathop{\mathrm{span}}\nolimits}
\newcommand{\Mat}{\mathop{\mathrm{M}}\nolimits}
\newcommand{\UT}{\mathop{\mathrm{UT}}\nolimits}
\newcommand{\DT}{\mathop{\mathrm{DT}}\nolimits}
\newcommand{\GL}{\mathop{\mathrm{GL}}\nolimits}
\newcommand{\spa}{\mathop{\mathrm{span}}\nolimits}
\newcommand{\supp}{\mathop{\mathrm{supp}}\nolimits}
\newcommand{\rank}{\mathop{\mathrm{rank}}\nolimits}
\newcommand{\idd}{\mathop{\mathrm{id}}\nolimits}
\newcommand{\ran}{\mathop{\mathrm{Ran}}\nolimits}
\newcommand{\dr}{ \mathop{\mathrm{d}_{{\mathbb R}^k}}\nolimits} 
\newcommand{\dc}{ \mathop{\mathrm{d}_{\cc}}\nolimits} \newcommand{\drr}{ \mathop{\mathrm{d}_{\rr}}\nolimits} 
\newcommand{\zin}{\mathbb{Z}}
\newcommand{\rr}{\mathbb{R}}
\newcommand{\cc}{\mathbb{C}}
\newcommand{\ww}{\mathbb{W}}
\newcommand{\nan}{\mathbb{N}}\newcommand{\bb}{\mathbb{B}}
\newcommand{\aaa}{\mathbb{A}}\newcommand{\ee}{\mathbb{E}}
\newcommand{\pp}{\mathbb{P}}
\newcommand{\wks}{\mathop{\mathrm{wk^*-}}\nolimits}
\newcommand{\mk}{{\Mat_k}}
\newcommand{\mnz}{\Mat_{n_0}}
\newcommand{\mn}{\Mat_{n}}
\newcommand{\dist}{\dc}
\newcommand{\braket}[2]{\left\langle#1,#2\right\rangle}
\newcommand{\ketbra}[2]{\left\vert #1\right \rangle \left\langle #2\right\vert}
\newcommand{\abs}[1]{\left\vert#1\right\vert}
\newtheorem{nota}{Notation}[section]
\def\qed{{\unskip\nobreak\hfil\penalty50
\hskip2em\hbox{}\nobreak\hfil$\square$
\parfillskip=0pt \finalhyphendemerits=0\par}\medskip}
\def\proof{\trivlist \item[\hskip \labelsep{\bf Proof.\ }]}
\def\endproof{\null\hfill\qed\endtrivlist\noindent}
\def\proofof[#1]{\trivlist \item[\hskip \labelsep{\bf Proof of #1.\ }]}
\def\endproofof{\null\hfill\qed\endtrivlist\noindent}
%%%%%%%%%%%%%%%%%%%%%%%%%%%%%%
\newcommand{\ZZ}{\bbZ_2\times\bbZ_2}
\newcommand{\SSS}{\mathcal{S}}
\newcommand{\cs}{S}
\newcommand{\ct}{t}
\newcommand{\hS}{S}
\newcommand{\vv}{{\boldsymbol v}}

\title{A $\bbZ_2$-index of symmetry protected topological phases with reflection symmetry for quantum spin chains}

\author{Yoshiko Ogata \thanks{ Graduate School of Mathematical Sciences
The University of Tokyo, Komaba, Tokyo, 153-8914, Japan
Supported in part by
the Grants-in-Aid for
Scientific Research, JSPS.}}
\maketitle

\begin{abstract}
For the classification of SPT phases, defining an index is a central problem.
In the famous paper \cite{po},
Pollmann, Tuner, Berg, and Oshikawa introduced $\bbZ_2$-indices for injective matrix products states 
(MPS) which have either $\bbZ_2\times \bbZ_2$ 
dihedral group (of $\pi$-rotations about $x$, $y$, and $z$-axes) symmetry,
time-reversal symmetry,
or reflection symmetry.
The first two are on-site symmetries. 
In \cite{Ogata4}, an index for on-site symmetries, which generalizes the
index in \cite{po}, was introduced
for general unique gapped ground state phases in quantum spin chains.
It was proved that the index is an invariant of the $C^1$-classification of SPT phases.
The index for the reflection symmetry, which is not an on-site symmetry, was left as an open question.
In this paper, we introduce a $\bbZ_2$-index for the reflection symmetric
unique gapped ground state phases, and complete the generalization problem of
index by Pollmann et.al.
We also show that the index is an invariant of the $C^1$-classification.

\end{abstract}
\section{Introduction}
Classification of unique gapped ground states in quantum many-body systems is an important problem in modern condensed matter physics and quantum information science.  In one dimension, it is believed that all unique gapped ground states belong to a single phase, in the sense that any two such ground states can be smoothly connected with each other thorough a series of models with unique gapped ground states.  This conjecture was verified for frustration free models with uniformly bounded degeneracy \cite{Ogata3}.  Motivated by the study of the Haldane phenomena in antiferromagnetic quantum spin chains, Gu and Wen  \cite{GuWen2009} proposed a finer classification based on the notion of symmetry protected topological (SPT) phase. Instead of considering the whole family of Hamiltonians,
we fix some symmetry and
consider the set of all Hamiltonians with a unique gapped ground state in the bulk, satisfying the symmetry.
%The classification problem of unique gapped ground state phases asks if two Hamiltonians are equivalent in the sense that
%they can be connected to each other via a continuous path of gapped Hamiltonians.
%Classification of SPT phase introduces some symmetry to the game.
%Let us fix some symmetry and consider the set of all Hamiltonians with a unique gapped ground state in the bulk, satisfying the symmetry.
We then say such two Hamiltonians are equivalent if
they can be connected to each other via a continuous path of symmetric Hamiltonians
with unique gapped ground state.
It can be possible that two symmetric Hamiltonians which can be connected via a path of
 {\it non-symmetric} gapped Hamiltonians
fails to be connected
via a path of {\it symmetric} gapped Hamiltonians.
A Hamiltonian
 which can not be connected to trivial Hamiltonians (i.e, Hamiltonians with on-site interactions) via a symmetry preserving path belongs to the SPT phase.
The question is how to show some Hamiltonian is in the SPT phase.
One way should be defining some index which is stable along the path of symmetric gapped Hamiltonians.
If some Hamiltonian has an index which is different from that of trivial phases,
the Hamiltonian should be in a SPT phase.
Finding such an index is a non-trivial important question for the classification problem of SPT phases.

In the famous paper \cite{po},
Pollmann, Tuner, Berg, and Oshikawa introduced $\bbZ_2$-indices for injective matrix products states 
(MPS) which have either $\bbZ_2\times \bbZ_2$ 
dihedral group (of $\pi$-rotations about $x$, $y$, and $z$-axes) symmetry,
time-reversal symmetry,
or reflection symmetry.
The first two are on-site symmetry, and the index is the cohomology class of
some projective representation associated to the symmetric injective MPS.
It was claimed there, that as the index takes discrete values,
it should be stable under the continuous path of gapped Hamiltonians.

The $\bbZ_2$-index beyond the framework of matrix product state was recently introduced by Tasaki
for systems satisfying on-site $U(1)$-symmetry together with one of $\bbZ_2\times\bbZ_2$-onsite symmetry/reflection symmetry/time reversal symmetry \cite{ta}.
He showed that these are actually invariant of the classification.

In \cite{Ogata4}, we extended the index of Pollmann et.al. for {\it on-site} symmetry with full generality 
(without asking $U(1)$-symmetry). 
We also proved  that our index is an invariant of the $C^1$-classification of SPT phases.
The index for the reflection symmetry, which is not an on-site symmetry, was left as an open question.
In this paper, we introduce a $\bbZ_2$-index for the reflection symmetric
unique gapped ground state phase, and complete the generalization problem of
index by Pollmann et.al.

Now let us state our result more in details. 
For a Hilbert space $\caH$, $B(\caH)$ denotes the set of all bounded operators on $\caH$.
If $V:\caH_1\to\caH_2$ is a linear/anti-linear map from a Hilbert space $\caH_1$ to 
another Hilbert space $\caH_2$,
then $\Ad (V):B(\caH_1)\to B(\caH_2)$ denotes the map
$\Ad(V)(x):=V x V^*$, $x\in B(\caH_1)$.

We start by summarizing standard setup of quantum spin chains on the infinite chain \cite{BR1,BR2}.
Throughout this paper, we fix some $2\le d\in\nan$.
We denote the algebra of $d\times d$ matrices by $\Mat_{d}$.
We denote the standard basis of $\cc^{d}$ by $\{\psi_\mu\}_{\mu=1,\ldots,d}$, and 
set $e_{\mu,\nu}=\ket{\psi_\mu}\bra{\psi_\nu}$ for each $\mu,\nu=1,\ldots,d$.

We denote the set of all finite subsets in ${\bbZ}$ by ${\mathfrak S}_{\bbZ}$.
For each $n\in\nan$, we set $\Lambda_n:=[-n,n]\cap \bbZ$.
%and the set of all finite intervals in ${\bbZ}$ by ${\mathfrak I}_{\bbZ}$.
%For each $X\in {\mathfrak S}_{\bbZ}$, $\diam(X)$ denotes the diameter of $X$.
%For $X,Y\subset \bbZ$, we denote by $d(X,Y)$, the distance between them.
%The number of elements in a finite set $\Lambda\subset {\bbZ}$ is denoted by
%$|\Lambda|$. For each $n\in\bbN$, we denote $[-n,n]\cap \bbZ$ by $\Lambda_n$.
%The complement of $\Lambda$ in $\bbZ$ is denoted by $\Lambda^c$.
%When we talk about intervals in $\bbZ$, $[a,b]$ for $a\le b$,
%means the interval in $\bbZ$, i.e., $[a,b]\cap \bbZ$.
%We denote the set of all finite intervals in $\Gamma$
%by ${\mathfrak I}_{\Gamma}$.
%We denote the set of all finite subsets in ${\bbZ}$ by ${\mathfrak S}_{\bbZ}$.
For each $z\in\bbZ$,  let $\caA_{\{z\}}$ be an isomorphic copy of $\Mat_{d}$, and for any finite subset $\Lambda\subset\bbZ$, let $\caA_{\Lambda} = \otimes_{z\in\Lambda}\caA_{\{z\}}$, which is the local algebra of observables in $\Lambda$. 
For finite $\Lambda$, the algebra $\caA_{\Lambda} $ can be regarded as the set of all bounded operators acting on
the Hilbert space $\otimes_{z\in\Lambda}{\bbC}^{d}$.
We use this identification freely.
If $\Lambda_1\subset\Lambda_2$, the algebra $\caA_{\Lambda_1}$ is naturally embedded in $\caA_{\Lambda_2}$ by tensoring its elements with the identity. 
The algebra $\caA_{R}$ (resp. $\caA_L$) representing the half-infinite chain
is given as the inductive limit of the algebras $\caA_{\Lambda}$ with $\Lambda\in{\mathfrak S}_{\bbZ}$, $\Lambda\subset[0,\infty)$
(resp. $\Lambda\subset(-\infty -1]$). 
The algebra $\caA$, representing the two sided infinite chain
is given as the inductive limit of the algebras $\caA_{\Lambda}$ with $\Lambda\in{\mathfrak S}_{\bbZ}$. 
Note that $\caA_{\Lambda}$ for $\Lambda\in {\mathfrak S}_{\bbZ}$,
$\caA_L$, and $\caA_R$ can be regarded naturally as subalgebras of
$\caA$.
We denote the set of local observables by $\caA_{\rm loc}=\bigcup_{\Lambda\in{\mathfrak S}_\bbZ}\caA_{\Lambda}
$.
We denote by $\beta_x$ the automorphisms on $\caA$ representing the space translation by  $x\in\bbZ$.
By $Q^{(j)}$, $j\in\bbZ$, we denote the element of $\caA$ with $Q\in\Mat_d$ in the $j$-th component of the tensor product of $\caA$ and the unit in any other component.
The reflection $\gamma$ is the unique  $*$-automorphism on $\caA$ 
which satisfies
\begin{align}
\gamma\lmk Q^{(j)}\rmk=Q^{(-j-1)},\quad \text{for all}\; Q\in\Mat_{d}\;\text{and}\; j \in\bbZ.
\end{align}
From $\gamma$, we define $*$-isomorphisms $\gamma_{R\to L}:\caA_R\to\caA_{L}$ and 
$\gamma_{L\to R}:\caA_L\to\caA_{R}$ by
\begin{align}
\gamma\lmk
\unit_{\caA_L}\otimes A
\rmk
=\gamma_{R\to L}(A)\otimes \unit_{\caA_R},\quad A\in\caA_R,
\end{align}
and 
\begin{align}
\gamma\lmk
B\otimes \unit_{\caA_R}
\rmk
=\unit_{\caA_L}\otimes\gamma_{L\to R}(B), \quad B\in\caA_L.
\end{align}

We introduce the $\bbZ_2$-index for reflection invariant  pure states satisfying the split property 
(see Definition \ref{split}) in Section \ref{indexsec} Definition \ref{index}. 
Since a unique gapped ground state of a reflection invariant Hamiltonians
satisfies these properties, this defines an index for
such systems. (See Section \ref{c1main}.)
The definition of the index is simple.
Let $\omega$ be a reflection invariant  pure state which satisfies the split property with respect to
$\caA_L$ and $\caA_R$. We then can find its GNS triple of the form
$(\caH_\omega\otimes \caH_\omega, \pi_\omega\circ\gamma_{L\to R}\otimes \pi_\omega, \Omega_\omega)$ (where $\pi_\omega$ is an irreducible representation of $\caA_R$)
and a unitary operator $\Gamma_\omega$ on $\caH_\omega\otimes \caH_\omega$ implementing 
$\gamma$.
(Lemma \ref{veq}.)
From the structure, we either have $\Gamma_\omega(\xi\otimes \eta)=\eta\otimes \xi$
for all $\xi,\eta\in\caH_\omega$
or $\Gamma_\omega(\xi\otimes \eta)=-\eta\otimes \xi$ for all $\xi,\eta\in\caH_\omega$.
(Theorem \ref{indint}.)
This sign $\sigma_\omega=\pm1$ is our $\bbZ_2$-index.
The same index 
can be obtained from
the Tomita-Takesaki modular conjugation.:
For the above GNS triple of $\omega$, 
let $\unit\otimes s_\omega$ be the support projection of $\Omega_\omega$ in
$\unit\otimes B(\caH_\omega)$.
Then we can consider modular conjugation $J_\omega$ associated to 
$s_\omega\otimes B(s_\omega \caH_\omega)$
 and $\Omega_\omega$.
%$J_\omega \lmk B(s_\omega\caH_\omega)\otimes s_\omega \rmk J_\omega^*=s_\omega\otimes B(s_\omega\caH_\omega)$.
%Then 
(Lemma \ref{thir}.)
There exists an anti-unitary
${\theta} :s_\omega\caH_\omega\to s_\omega\caH_\omega$ such that
\begin{align}
J_\omega\lmk s_\omega\otimes x\rmk J_\omega^*
=\theta x\theta^*\otimes s_\omega,\quad J_\omega\lmk x\otimes s_\omega\rmk J_\omega^*
=s_\omega\otimes \theta x\theta^*,
\end{align}
for all $x\in B(s_\omega\caH_\omega)$,
(Proposition \ref{kdef}.)
%As the modular operator $J_\omega$ satisfies $J_\omega^2=s_\omega\otimes s_\omega$,
This ${\theta}$ satisfies ${\theta}^2=\kappa_\omega s_\omega$ with
some $\kappa_\omega\in\{-1,1\}$, because of $J_\omega^2=s_\omega\otimes s_\omega$. 
It turns out that $\kappa_\omega$ coincides with our $\bbZ_2$-index $\sigma_\omega$. 
(Theorem \ref{coffee}.)
This $\theta$ is related to the Schmidt decomposition of $\Omega_\omega$.(Lemma \ref{thir}.)
Therefore, considering the Schmidt decomposition can be one way to calculate the index $\sigma_\omega$.
(Remark \ref{su}.)
%we have $\theta_\omega^2$
%that $\Omega_\omega$ is cyclic and separating for $s_\omega\otimes B(s_\omega \caH_\omega)$.Therefore, we may consider the modular conjugation on 
% $J_{\omega}$
%for $\Omega_\omega$, and $\caM_\omega:=$. 

As stated above, for reflection invariant injective matrix product states,
a $\bbZ_2$-index was introduced in \cite{po}.
It turns out that our $\bbZ_2$-index restricted to such states
coincides with that of \cite{po}.
This is proven in Section \ref{poco} using the relation $\kappa_\omega=\sigma_\omega$.

The $\bbZ_2$-index $\sigma_\omega$ is invariant under 
automorphic equivalence via an automorphism
which allows a reflection invariant decomposition.:
\begin{idefn}
We say an automorphism $\alpha$ of $\caA$ allows a reflection invariant 
decomposition if there is an automorphisms $\alpha_R$
on $\caA_R$, and a unitary $W$ in $\caA$
such that
\begin{align}
\tilde\alpha^{-1}\circ\alpha=\Ad(W),\quad \gamma(W)=W,
\end{align}
where
\begin{align}
\tilde\alpha:=\lmk
\gamma_{R\to L} \circ\alpha_R\circ \gamma_{L\to R}\rmk\otimes \alpha_R.
\end{align}
%We call $(\alpha_R,W)$, a reflection invariant decomposition of $\alpha$.
\end{idefn}
From the definition, we can show the following:
\begin{ithm}(See Theorem \ref{aeii} for details.)
Let $\omega_0,\omega_1$ be reflection invariant pure states satisfying the split property.
Suppose that $\omega_0$ and $\omega_1$ 
are automorphic equivalent via an automorphism
which allows a reflection invariant decomposition.
Then the $\bbZ_2$-indices $\sigma_{\omega_0}$,  $\sigma_{\omega_1}$ associated to
$\omega_0$, $\omega_1$ are equal.
\end{ithm}

Recalling that a unique gapped ground state is pure and satisfies the split property (see Theorem \ref{matsui}), our $\bbZ_2$-index can be understood as an index of 
of reflection invariant Hamiltonians with unique gapped ground states.
It turns out that this $\bbZ_2$-index is an invariant of the $C^1$-classification.:
\begin{icor}(See Theorem \ref{c1t} for more precise statement)
Let us consider 
a $C^1$-path of interactions, 
in the reflection invariant unique gapped ground state phase.
Suppose that if we associate some suitable boundary conditions along the path, they give
 local Hamiltonians which are gapped for an
increasing sequence of finite boxes. (See Definition \ref{boundary}.)
Then the $\bbZ_2$-index $\sigma_\omega$ does not change along the path.
\end{icor}
This can be shown from the fact that ground states along the $C^1$-path
are mutually automorphic equivalent via an automorphism which allows a reflection invariant decomposition.
The boundary conditions in the Corollary can be arbitrary, as long as they guarantee the gap. We may take it as periodic boundary condition, for example.
Furthermore, the boundary condition itself does not need to be reflection invariant.

Our theorem, along with results in [PTOB1,PTOB2,CGW,Tas2] about matrix product states, shows that 
AKLT interaction and trivial interaction belong to different reflection symmetric unique gapped ground state phases.
In other words, 
AKLT interaction and trivial interaction can never be connected by a $C^1$-path of reflection invariant 
interactions without without closing the gap.

\section{The $\bbZ_2$-index associated to the reflection symmetric split states}\label{indexsec}
We introduce $\bbZ_2$-index for reflection invariant pure state satisfying the split property.
Let us first recall the definition of the split property.
Here we give the following definition, which is most suitable for our purpose.
It corresponds to the standard definition \cite{dl} in our setting (see \cite{Matsui2}).
%See Definition~1.1 of \cite{Matsui2} for a different characterization, which is more physically motivated.
\begin{defn}\label{split}
Let $\varphi$ be a pure state on $\caA$. Let $\varphi_R$ be the restriction of
$\varphi$ to $\caA_R$, and $(\caH_{\varphi_R},\pi_{\varphi_R},\Omega_{\varphi_R})$ be the GNS triple of $\varphi_R$.
We say $\varphi$ satisfies the split property with respect to $\caA_L$ and $\caA_R$,
if the von Neumann algebra $\pi_{\varphi_R}(\caA_{R})''$ is a type I factor.
\end{defn}
Recall that a type I factor is isomorphic to $B(\caK)$, the set of all bounded operators 
on a  Hilbert space $\caK$. 
See \cite{takesaki}.
We consider following type of  GNS-triple for a reflection invariant pure state which satisfies the split property.
\begin{defn}\label{rsr}
Let $\omega$ be
a reflection invariant pure state
on $\caA$ which satisfies the split property with respect to $\caA_R$ and
$\caA_L$. 
We say $(\caH_\omega, \pi_\omega, \Omega_\omega, \Gamma_\omega)$  is a reflection-spilt representation associated to
$\omega$ if setting $\hat \caH_\omega:=\caH_\omega\otimes \caH_\omega$ and
$\hat\pi_\omega:=\lmk \pi_\omega\circ\gamma_{L\to R}\rmk\otimes \pi_\omega$, following hold:
\begin{enumerate}
\item $(\caH_\omega, \pi_\omega)$ is an irreducible representation of $\caA_R$, 
\item  $\Omega_\omega$ is a unit vector of $\hat \caH_\omega$,
\item  $(\hat \caH_\omega, 
\hat\pi_\omega, \Omega_\omega)$ is a GNS triple of $\omega$,
\item $\Gamma_\omega$ is the unique unitary operator on $\hat \caH_\omega$ such that
\begin{align}\label{guni}
\Gamma_\omega \hat\pi_\omega(A)\Omega_\omega=\hat\pi_\omega\circ\gamma(A)\Omega_\omega,\quad A\in\caA.
\end{align}
\end{enumerate}
\end{defn}
\begin{rem}\label{gm2}
Because of $\gamma^2=\id$, from the definition of $\Gamma_\omega$  (\ref{guni}),
we have $\Gamma_\omega^2=\unit_{\caH_\omega}$. 
\end{rem}
\begin{rem}
For the rest of this paper, for any reflection-spilt representation 
$(\caH_\omega, \pi_\omega, \Omega_\omega, \Gamma_\omega)$
we use the notation $\hat \caH_\omega:=\caH_\omega\otimes \caH_\omega$ and
$\hat\pi_\omega:=\lmk \pi_\omega\circ\gamma_{L\to R}\rmk\otimes \pi_\omega$.
\end{rem}

\begin{lem}\label{veq}
For
any reflection invariant pure state $\omega$ on $\caA$ which satisfies the split property with respect to $\caA_R$ and
$\caA_L$, there exists 
 a reflection-spilt representation $(\caH_\omega, \pi_\omega, \Omega_\omega, \Gamma_\omega)$  associated to
 $\omega$.
 Furthermore, if $(\bar \caH_\omega, \bar \pi_\omega, \bar\Omega_\omega, \bar\Gamma_\omega)$
is another reflection-spilt representation of $\omega$,
there exists a unitary $V:\caH_\omega\to \bar\caH_\omega$ such that
such that 
\begin{align}
&\Ad(V)\circ\pi_\omega(A)
=\bar\pi_\omega(A),\quad A\in\caA_R,\label{fst}\\
&\bar\Gamma_\omega=\Ad \lmk V\otimes V\rmk\lmk \Gamma_\omega \rmk,\label{scd}\\
&(V\otimes V)\Omega_\omega=\bar \Omega_\omega\label{trd}.
\end{align}
\end{lem}
\begin{proof}
Let $(\caH_R, \pi_R,\Omega_R)$ be a GNS triple of $\omega\vert_{\caA_R}$.
Note that from the reflection invariance of $\omega$, 
 $(\caH_R, \pi_R\circ\gamma_{L\to R},\Omega_R)$ is a GNS triple 
 of $\omega\vert_{\caA_L}$.
 Therefore, 
 $(\caH_R\otimes \caH_R, \pi_R\circ\gamma_{L\to R}\otimes \pi_R, \Omega_R\otimes\Omega_R)$
 is a GNS triple of $\omega\vert_{\caA_L}\otimes \omega\vert_{\caA_R}$.
 
 Since $\omega$ satisfies the split property, there exists a Hilbert space 
$\caH_\omega$ and a $*$-isomorphism $\iota:\pi_R(\caA_R)''\to B(\caH_\omega)$.
We introduce a representation $\pi_\omega:=\iota\circ\pi_R$ of $\caA_R$ on $\caH_\omega$.
Since $\iota$ is a $*$-isomorphism, $(\caH_\omega,\pi_\omega)$ is an irreducible representation of $\caA_R$.

Set $\hat \caH_\omega:=\caH_\omega\otimes \caH_\omega$ and let $\hat\pi_\omega:=\lmk \pi_\omega\circ\gamma_{L\to R}\rmk\otimes 
\pi_\omega$ be the representation of
$\caA$ on $\hat \caH_\omega$.
Now we would like to show the existence of a unit vector 
$\Omega_\omega\in \hat\caH_\omega$ 
such that $(\hat\caH_\omega, \hat\pi_\omega, \Omega_\omega)$
is a GNS triple of $\omega$.
Since there is a $*$-isomorphism
 \[
 \iota\otimes \iota:\lmk \pi_R\circ\gamma_{L\to R}(\caA_L)\rmk''\otimes \pi_R(\caA_R)''=
 \lmk \pi_R(\caA_R)\rmk''\otimes \pi_R(\caA_R)''\to 
 B(\caH_\omega\otimes \caH_\omega)
 \]
 (Theorem 5.2 IV of \cite{takesaki}),
the representation $\hat\pi_\omega=\lmk
\iota\otimes \iota \rmk\circ \lmk\pi_R\circ\gamma_{L\to R}\otimes \pi_R\rmk$
is quasi-equivalent to $\pi_R\circ\gamma_{L\to R}\otimes \pi_R$,
the GNS representation of $\omega\vert_{\caA_L}\otimes \omega\vert_{\caA_R}$.
(Theorem 2.4.26 \cite{BR1}.)
On the other hand, as $\omega$ satisfies the split property, 
by the proof of Proposition 2.2 of \cite{Matsui1},
$\omega\vert_{\caA_L}\otimes \omega\vert_{\caA_R}$ is quasi-equivalent to $\omega$. ( In Proposition 2.2 of \cite{Matsui1}, it is assumed that the state is translationally invariant because of the first equivalent condition (i). However, the proof for the equivalence between
(ii) and (iii) does not require translation invariance.)
Hence $\hat\pi_\omega$ is quasi-equivalent to
the GNS representation of $\omega$. (See section 2.4 of \cite{BR1}.)
Therefore, there is a density matrix $\rho$ on $\hat\caH_\omega$ 
such that
\begin{align}
\omega(A)
=\Tr\rho\lmk
\hat \pi_\omega(A)
\rmk,\quad A\in\caA.
\end{align}
Since $\omega$ is pure and $\hat\pi_\omega(\caA)''=B(\hat\caH_\omega)$,
this $\rho$ should be a one rank operator onto a one dimensional space $\bbC\Omega_\omega$,
with some unit vector $\Omega_\omega\in\hat\caH_\omega$.
Because of $\hat\pi_\omega(\caA)''=B(\hat\caH_\omega)$,
$\Omega_\omega$ is cyclic for
$\hat\pi_\omega(\caA)$, and $(\hat\caH_\omega, \hat\pi_\omega, \Omega_\omega)$
is a GNS triple of $\omega$. From the $\gamma$-invariance of $\omega$, there exists $\Gamma_\omega$ satisfying
(\ref{guni}).
(Corollary 2.3.17 of \cite{BR1}.)

Now let us show the uniqueness. Let $(\bar \caH_\omega, \bar \pi_\omega, \bar\Omega_\omega, \bar\Gamma_\omega)$
be another reflection-spilt representation of $\omega$.
Since both of $(\caH_\omega\otimes \caH_\omega, \lmk\pi_\omega\circ\gamma_{L\to R}\rmk\otimes 
\pi_\omega, \Omega_\omega)$
and $(\bar \caH_\omega\otimes \bar \caH_\omega, \lmk\bar \pi_\omega\circ\gamma_{L\to R}\rmk\otimes 
\bar\pi_\omega, \bar\Omega_\omega)$
are GNS triple of $\omega$, there is a unitary $U:\caH_\omega\to \bar\caH_\omega$
such that
\begin{align}\label{gnsu}
U\lmk \pi_\omega\circ\gamma_{L\to R}\otimes \pi_\omega\rmk\lmk A \rmk U^*
= \lmk {\bar \pi}_\omega\circ\gamma_{L\to R}\otimes \bar\pi_\omega\rmk(A),\quad A\in\caA,\;\text{and}\;
U\Omega_\omega=\bar\Omega_\omega.
\end{align}
(Theorem 2.3.16 of \cite{BR1}.)
Restricting the first equation to $\caA_R$, we have
\begin{align}
U\lmk \unit_{\caH_\omega}\otimes \pi_\omega\lmk A \rmk \rmk U^*
=\unit_{\caH_\omega}\otimes \bar \pi_\omega(A),\quad A\in\caA_R.
\end{align}
From this we obtain a $*$-isomorphism $\tau$ from 
$B(\caH_\omega)=\pi_\omega(\caA_R)''$ to $B(\bar\caH_\omega)=\bar\pi_\omega(\caA_R)''$
such that
\begin{align}
U\lmk \unit_{\caH_\omega}\otimes x\rmk U^*
=\unit_{\bar\caH_\omega}\otimes \tau(x),\quad x\in\caB(\caH_\omega),
\end{align}
which satisfies
\begin{align}
\tau\circ \pi_\omega\lmk A \rmk 
= \bar \pi_\omega(A),\quad A\in\caA_R.
\end{align}
Applying Wigner's theorem to $\tau$, there exists a unitary $\tilde V:\caH_\omega\to\bar\caH_\omega$
such that
\begin{align}
\tau(x)=\tilde Vx\tilde V^*,\quad x\in B(\caH_\omega).
\end{align}
In particular, we have
\begin{align}\label{vact}
\Ad(\tilde V)\lmk \pi_\omega(A)\rmk=\bar\pi_\omega(A),\quad A\in\caA_R.
\end{align}
From this and (\ref{gnsu}), we have
\begin{align}
&\Ad(U)\lmk \lmk  \pi_\omega\circ\gamma_{L\to R}\otimes \pi_\omega\rmk\lmk A \rmk \rmk
= \lmk {\bar \pi}_\omega\circ\gamma_{L\to R}\otimes \bar\pi_\omega\rmk(A)\nonumber\\
&=\Ad(\tilde V\otimes \tilde V)\lmk \lmk  \pi_\omega\circ\gamma_{L\to R}\otimes \pi_\omega\rmk\lmk A \rmk\rmk 
,\quad A\in\caA.
\end{align}
Since $\lmk \lmk  \pi_\omega\circ\gamma_{L\to R}\otimes \pi_\omega\rmk\lmk \caA \rmk \rmk''=
B(\caH_\omega\otimes\caH_\omega)$, 
$(\tilde V\otimes \tilde V)^*U=c\unit_{\caH_\omega}$, for some $c\in \bbT$.
Choosing one $c_1\in \bbT$ such that $c_1^2=c$,
we set $V:=c_1\tilde V$.
Then we have $U=V\otimes V$, and
from (\ref{gnsu}), (\ref{trd}) holds.
The property (\ref{scd}) holds from (\ref{gnsu}) and (\ref{trd}) .
By (\ref{vact}), we obtain (\ref{fst}).
\end{proof}
For a reflection invariant pure state which satisfies the split property with respect to $\caA_R$ and
$\caA_L$ we can define an index via the reflection-split representation of $\omega$.
\begin{thm}\label{indint}
Let $\omega$ be
a reflection invariant pure state which satisfies the split property with respect to $\caA_R$ and
$\caA_L$. 
Let $(\caH_\omega, \pi_\omega, \Omega_\omega, \Gamma_\omega)$  be a reflection-spilt representation associated to
$\omega$.
Then we have
\begin{align}\label{gxg}
\Gamma_\omega\lmk \unit_{\caH_\omega}\otimes x\rmk\Gamma_\omega^*
=x\otimes \unit_{\caH_\omega},\quad
x\in B(\caH_\omega).
\end{align}
Furthermore, there exists a constant $\sigma_\omega=\pm 1$
such that
\begin{align}\label{gsd}
\Gamma_\omega\lmk \xi\otimes \eta\rmk
=\sigma_{\omega}\eta\otimes \xi,\quad
\text{for all} \quad \xi,\eta\in\caH_\omega.
\end{align}
This $\sigma_\omega$ is independent of the choice of the reflection-spilt representation
$(\caH_\omega, \pi_\omega, \Omega_\omega, \Gamma_\omega)$.
\end{thm}
\begin{defn}\label{index}
From Theorem \ref{indint}, for each
 reflection invariant pure state $\omega$ which satisfies the split property with respect to $\caA_R$ and
$\caA_L$
, we can define a $\bbZ_2$-index $\sigma_\omega$.
\end{defn}
\begin{proof}
We first prove (\ref{gxg}).
For any $A\in \caA_R$, we have 
\begin{align}
&\Gamma_\omega\lmk
\unit_{\caH_\omega}\otimes \pi_\omega(A)
\rmk\Gamma_\omega^*
=\Gamma_\omega\lmk
\hat\pi_\omega\lmk
\unit_{\caA_L}\otimes A
\rmk\rmk
\Gamma_\omega^*
=\hat\pi_\omega\circ \gamma\lmk
\unit_{\caA_L}\otimes A
\rmk\nonumber\\
&=\hat\pi_\omega
\lmk
\gamma_{R\to L}(A)\otimes \unit_{\caH_\omega}
\rmk
=\pi_\omega\circ\gamma_{L\to R}\lmk\gamma_{R\to L}(A)\rmk\otimes \unit_{\caA_R}
=\pi_\omega(A)\otimes \unit_{\caH_\omega}.
\end{align}
Since both sides are $\sigma$-weak continuous and $\pi_\omega(\caA_R)''=B(\caH_\omega)$, we obtain
 (\ref{gxg}).

From (\ref{gxg}), we derive (\ref{gsd}).:
For any nonzero $\xi,\eta\in\caH_\omega$,
there exists $\sigma_{\xi,\eta}\in\bbT$ such that
\begin{align}\label{xe}
\Gamma_\omega\lmk\xi\otimes\eta\rmk
=\sigma_{\xi,\eta} \lmk\eta\otimes\xi\rmk,
\end{align}
because
\begin{align}
\Gamma_\omega
\ketbra{\xi\otimes\eta}{\xi\otimes\eta}
\Gamma_\omega^*
=\ketbra{\eta\otimes\xi}{\eta\otimes\xi},
\end{align}
from (\ref{gxg}).
Considering the case $\xi=\eta\neq 0$ in (\ref{xe}), we have
\begin{align}
\xi\otimes\xi=
\Gamma_\omega^2\lmk\xi\otimes\xi\rmk
=\sigma_{\xi,\xi} \Gamma_\omega\lmk\xi\otimes\xi\rmk
=\sigma_{\xi,\xi} ^2\xi\otimes\xi.
\end{align}
The first equality is from Remark \ref{gm2}. 
From this, we obtain $\sigma_{\xi,\xi}=\pm 1$.
Again by (\ref{gxg}), for nonzero $\xi,\eta\in\caH_\omega$, we obtain
\begin{align}
\Gamma_\omega\ketbra{\xi\otimes \xi}{\eta\otimes\eta}\Gamma_\omega^*
=\ketbra{\xi\otimes \xi}{\eta\otimes\eta}.
\end{align}
On the other hand, from the above argument, we have
\begin{align}
\Gamma_\omega\ketbra{\xi\otimes \xi}{\eta\otimes\eta}\Gamma_\omega^*
=\sigma_{\xi,\xi}\sigma_{\eta,\eta} \ketbra{\xi\otimes \xi}{\eta\otimes\eta}.
\end{align}
Since $\eta,\xi$ are not zero, we obtain $\sigma_{\xi,\xi}\sigma_{\eta,\eta} =1$.
Recalling that $\sigma_{\xi,\xi}, \sigma_{\eta,\eta}$ take values in $\pm 1$, we obtain
$\sigma_{\xi,\xi}=\sigma_{\eta,\eta}$.
Therefore, we set $\sigma_\omega:=\sigma_{\xi,\xi}$,
which is independent of the choice of nonzero $\xi\in\caH_\omega$.
To prove (\ref{gsd}), we use (\ref{gxg}) again and for any nonzero $\xi,\eta\in\caH_\omega$, we have
\begin{align}
\ketbra{\eta\otimes \xi}{\xi\otimes\xi}
=\Gamma_\omega\ketbra{\xi\otimes \eta}{\xi\otimes\xi}\Gamma_\omega^*
=\sigma_{\xi,\eta}\sigma_{\omega} \ketbra{\eta\otimes \xi}{\xi\otimes\xi}.
\end{align}
The first equality follows from (\ref{gxg}) and the second one is the definition 
of $\sigma_{\xi,\eta}$ and $\sigma_{\omega}$.
From this and $\sigma_\omega=\pm1$, we obtain $\sigma_{\xi,\eta}=\sigma_{\omega}$, completing the proof of (\ref{gsd}).

To show  that the sign $\sigma_\omega$ is independent of the choice
of the reflection-split representation, let 
$(\bar\caH_\omega, \bar\pi_\omega, \bar \Omega_\omega, \bar\Gamma_\omega)$  be another reflection-spilt representations associated to
$\omega$.
Let $V:\caH_\omega\to\bar\caH_\omega$ be the unitary given in Lemma \ref{veq}.
We have
\begin{align}
\bar\Gamma_\omega\lmk\xi\otimes\eta\rmk
=\lmk V\otimes V\rmk\Gamma_\omega\lmk V^*\xi\otimes V^*\eta\rmk
=\lmk V\otimes V\rmk\sigma_\omega\lmk V^*\eta\otimes V^*\xi\rmk
=\sigma_\omega\lmk \eta\otimes \xi\rmk,
\end{align}
for any $\xi,\eta\in \bar\caH_\omega$,
proving the claim.
\end{proof}

Now we prove that the $\bbZ_2$-index is invariant under 
automorphic equivalence via an automorphism
which allows a reflection invariant decomposition.
Let us recall the definition:
\begin{defn}\label{adt}
We say an automorphism $\alpha$ of $\caA$ allows a reflection invariant 
decomposition if there is an automorphisms $\alpha_R$
on $\caA_R$, and a unitary $W$ in $\caA$
such that
\begin{align}
\tilde\alpha^{-1}\circ\alpha=\Ad(W),\quad \gamma(W)=W,
\end{align}
where
\begin{align}
\tilde\alpha:=\lmk
\gamma_{R\to L} \circ\alpha_R\circ \gamma_{L\to R}\rmk\otimes \alpha_R.
\end{align}
We call $(\alpha_R,W)$, a reflection invariant decomposition of $\alpha$.
\end{defn}
We prove the following theorem:
\begin{thm}\label{aeii}
Let $\omega_0,\omega_1$ be reflection invariant pure states satisfying the split property
with respect to $\caA_R$ and
$\caA_L$.
Suppose that $\omega_0$ and $\omega_1$ 
are automorphic equivalent via an automorphism $\alpha$, i.e., $\omega_1=\omega_0\circ\alpha$,
which allows a reflection invariant decomposition  $(\alpha_R,W)$.
Then the $\bbZ_2$-indices $\sigma_{\omega_0}$,  $\sigma_{\omega_1}$ associated to
$\omega_0$, $\omega_1$ are equal.
\end{thm}
\begin{proof}
Let $(\caH_{\omega_0}, \pi_{\omega_0}, \Omega_{\omega_0}, \Gamma_{\omega_0})$  be a reflection-spilt representation associated to
$\omega_0$.
Set $\hat \caH_{\omega_0}:=\caH_{\omega_0}\otimes \caH_{\omega_0}$,
$\hat\pi_{\omega_0}:=\lmk \pi_{\omega_0}\circ\gamma_{L\to R}\rmk\otimes \pi_{\omega_0}$.
We also set $\tilde\alpha:=\lmk
\gamma_{R\to L} \circ\alpha_R\circ \gamma_{L\to R}\rmk\otimes \alpha_R$ and
\[
\hat\pi_{\omega_1}:=\lmk \pi_{\omega_0}\circ\alpha_R\circ\gamma_{L\to R}\rmk\otimes 
\lmk \pi_{\omega_0}\circ
\alpha_R\rmk
=\hat\pi_{\omega_0}\circ\tilde\alpha.
\]
We claim that $(\caH_{\omega_0}, \pi_{\omega_0}\circ\alpha_R, 
\hat\pi_{\omega_1}(W^*)\Omega_{\omega_0}, \Gamma_{\omega_0})$
is a reflection-split representation of $\omega_1$.
From this, we obtain the statement of the Theorem i.e., $\sigma_{\omega_0}=\sigma_{\omega_1}$.

The first condition of Definition \ref{rsr} is from $\pi_{\omega_0}\circ\alpha_R(\caA_R)''=\pi_{\omega_0}(\caA_R)''=B(\caH_{\omega_0})$.
The second one is trivial because $W$ is unitary.
To prove the third one, note that $(\hat \caH_{\omega_0}, \hat\pi_{\omega_1}, \Omega_{\omega_0})$ is a GNS triple of 
$\omega_0\circ \tilde\alpha$.
From $\tilde\alpha^{-1}\circ\alpha=\Ad(W)$ and $\omega_1=\omega_0\circ\alpha$,
we have
\begin{align}
\omega_1=\omega_0\circ\alpha=\omega_0\circ \tilde\alpha\circ\Ad(W).
\end{align}
Combining these two facts, we have
\begin{align}
\omega_1(A)=\braket{\hat\pi_{\omega_1}(W^*)\Omega_{\omega_0}}{\hat\pi_{\omega_1}(A)\hat\pi_{\omega_1}(W^*)\Omega_{\omega_0}},\quad A\in\caA.
\end{align}
Since $\hat\pi_{\omega_1}(\caA)''=B(\hat\caH_{\omega_0})$,
$\hat\pi_{\omega_1}(W^*)\Omega_{\omega_0}$ is cyclic for $\hat\pi_{\omega_1}(\caA)$.
Hence $(\hat \caH_{\omega_0}, \hat\pi_{\omega_1}, \hat\pi_{\omega_1}(W^*)\Omega_{\omega_0})$ is a GNS triple of 
$\omega_1$.
Finally for the fourth condition of Definition \ref{rsr},
note that $\gamma$ and $\tilde\alpha$ commute because
\begin{align}
&\gamma\circ\tilde \alpha(A\otimes B)
=\gamma\circ\lmk \lmk
\gamma_{R\to L} \circ\alpha_R\circ \gamma_{L\to R}(A)\rmk\otimes \alpha_R(B)\rmk
= \lmk \gamma_{R\to L}\circ\alpha_R(B)\rmk\otimes\lmk \alpha_R\circ \gamma_{L\to R}(A)\rmk\nonumber\\
&=\lmk
\gamma_{R\to L} \circ\alpha_R\circ \gamma_{L\to R}\circ \gamma_{R\to L}(B)\rmk\otimes 
\alpha_R\circ\gamma_{L\to R}(A)
=\tilde\alpha\circ\gamma(A\otimes B),
\end{align}
for any $A\in\caA_L$ and $B\in \caA_R$.
From this fact, we obtain
\begin{align}
&\Gamma_{\omega_0}\hat\pi_{\omega_1}\lmk A\rmk\hat\pi_{\omega_1}(W^*)\Omega_{\omega_0}
%=\Gamma_{\omega_0}\hat\pi_{\omega_0}\circ \tilde \alpha\lmk A\rmk\hat\pi_{\omega_1}(W^*)\Omega_{\omega_0}
=\Gamma_{\omega_0}\hat\pi_{\omega_0}\circ \tilde \alpha\lmk AW^*\rmk
\Omega_{\omega_0}
=\lmk\hat\pi_{\omega_0}\circ\gamma\circ \tilde \alpha\lmk AW^*\rmk\rmk
\Omega_{\omega_0}\nonumber\\
&=\lmk\hat\pi_{\omega_0}\circ \tilde \alpha\circ\gamma\lmk AW^*\rmk\rmk
\Omega_{\omega_0}
=\lmk
\hat\pi_{\omega_0}\circ \tilde \alpha\circ\gamma\lmk A\rmk\rmk
\lmk
\hat\pi_{\omega_0}\circ \tilde \alpha\lmk W^*\rmk\rmk
\Omega_{\omega_0}
=\lmk
\hat\pi_{\omega_1}\circ\gamma\lmk A\rmk\rmk
\lmk
\hat\pi_{\omega_1}\lmk W^*\rmk\rmk
\Omega_{\omega_0},
\end{align}
for all $A\in\caA$.
For the fourth equality, we used $\gamma(W)=W$.
This completes the proof.
\end{proof}

\section{$C^1$-classification of gapped Hamiltonians with the reflection symmetry.}\label{c1main}
Let us now apply the result in Section \ref{indexsec} to the $C^1$-classification of gapped Hamiltonians preserving the reflection symmetry. 

A mathematical model of a quantum spin chain is fully specified by its interaction $\Phi$.
An interaction is a map $\Phi$ from 
${\mathfrak S}_{\bbZ}$ into ${\caA}_{\rm loc}$ such
that $\Phi(X) \in {\caA}_{X}$ 
and $\Phi(X) = \Phi(X)^*$
for each $X \in {\mathfrak S}_{\bbZ}$. 
Let $R:\bbZ\to\bbZ$ be the reflection : $R(i):=-i-1$, $i\in\bbZ$.
An interaction $\Phi$ is reflection invariant
if $\gamma(\Phi(X))=\Phi(R(X))$
for all $X\in {\mathfrak S}_\bbZ$.
An interaction $\Phi$
 is of finite range if there exists an $m\in {\mathbb N}$ such that
$\Phi(X)=0$ for $X$ with diameter larger than $m$.
%In this case, we say $\Phi$ has an interaction length less than or equal to $m$.
We denote by $\caB_{f}$, 
the set of all finite range interactions $\Phi$ which satisfy
\begin{align}\label{fi}
a_\Phi:= \sup_{X\in {\mathfrak S}_{\bbZ}}
\lV
\Phi\lmk X\rmk
\rV<\infty.
\end{align}
%An interaction $\Phi$ is translation invariant if
%$
%\Phi(X+j)=\tau_j\lmk
%\Phi(X)\rmk,
%$
%for all $ j\in{\mathbb Z}$ and $X\in  {\mathfrak S}_{\bbZ}$.
We may define addition on $\caB_f$: for $\Phi,\Psi\in\caB_f$,
$\Phi+\Psi$ denotes the interaction given by
$(\Phi+\Psi)(X)=\Phi(X)+\Psi(X)$ for each $X\in{\mathfrak S}_{\bbZ}$.

For an interaction $\Phi$ and a finite set $\Lambda\in{\mathfrak S}_{\bbZ}$, we define the local Hamiltonian on $\Lambda$ by
\begin{equation}\label{GenHamiltonian}
\lmk H_{\Phi}\rmk_{\Lambda}:=\sum_{X\subset{\Lambda}}\Phi(X).
\end{equation}
The dynamics given by this local Hamiltonian is denoted by
\begin{align}\label{taulamdef}
\tau_t^{\Phi,\Lambda}\lmk A\rmk:= e^{it\lmk H_{\Phi}\rmk_{\Lambda}} Ae^{-it\lmk H_{\Phi}\rmk_{\Lambda}},\quad
t\in\bbR,\quad A\in\caA.
\end{align}
%We consider the following boundedness condition of translation invariant interactions
%\begin{align}\label{bdd}
%\sum_{X\ni 0}\frac{\lV \Phi(X)\rV}{|X|}<\infty
%\end{align}
If $\Phi$ belongs to $\caB_{f}$,
the limit
\begin{align}\label{taudef}
\tau_t^{\Phi}\lmk A\rmk=\lim_{\Lambda\to\bbZ}
\tau_t^{\Phi,\Lambda}\lmk A\rmk
\end{align}
% e^{it\lmk H_{\Phi}\rmk_{\Lambda}} Ae^{-it\lmk H_{\Phi}\rmk_{\Lambda}}
exists for each $A\in \caA$ and $t\in{\mathbb R}$, 
and defines a strongly continuous one parameter group of automorphisms $\tau^\Phi$ on $\caA$. 
(See \cite{BR2}.)
We denote the generator of the $C^* $-dynamics $\tau^{\Phi}$ by $\delta_{\Phi}$.

%In addition to (\ref{taulamdef}) and (\ref{taudef}), we introduce the following dynamics 
%corresponding to the left-right decomposition. For $\Phi\in \caB_f$, we set
%\begin{align}
%&\tilde\tau_t^{\Phi,\Lambda}\lmk A\rmk
%:=e^{it
%\lmk
%\lmk H_{\Phi}\rmk_{\Lambda\cap(-\infty,-1]}+\lmk H_{\Phi}\rmk_{\Lambda\cap([0,\infty)}
%\rmk}
% A
%  e^{-it
%\lmk
%\lmk H_{\Phi}\rmk_{\Lambda\cap(-\infty,-1]}+\lmk H_{\Phi}\rmk_{\Lambda\cap([0,\infty)}
%\rmk},
%%&\tilde\tau_t^{\Phi(s)}\lmk A\rmk
%%=\lim_{\Lambda\to\bbZ} \tilde\tau_t^{\Phi(s),\Lambda}\lmk A\rmk.
%\end{align}
%and
%\begin{align}\label{ttaudef}
%\tilde \tau_t^{\Phi}\lmk A\rmk=\lim_{\Lambda\to\bbZ}
%\tau_t^{\tilde \Phi,\Lambda}\lmk A\rmk,\quad A\in\caA.
%\end{align}

For $\Phi\in\caB_f$, a state $\varphi$ on $\caA$ is called a \mbox{$\tau^{\Phi}$-ground} state
if the inequality
$
-i\,\varphi(A^*{\delta_{\Phi}}(A))\ge 0
$
holds
for any element $A$ in the domain $\caD({\delta_{\Phi}})$ of ${\delta_\Phi}$.
Let $\varphi$ be a $\tau^\Phi$-ground state, with the GNS triple $(\caH_\varphi,\pi_\varphi,\Omega_\varphi)$.
Then there exists a unique positive operator $H_{\varphi,\Phi}$ on $\caH_\varphi$ such that
$e^{itH_{\varphi,\Phi}}\pi_\varphi(A)\Omega_\varphi=\pi_\varphi(\tau_t^\Phi(A))\Omega_\varphi$,
for all $A\in\caA$ and $t\in\mathbb R$.
We call this $H_{\varphi,\Phi}$ the bulk Hamiltonian associated with $\varphi$.
Note that $\Omega_\varphi$ is an eigenvector of $H_{\varphi,\Phi}$ with eigenvalue $0$. See \cite{BR2} for the general theory.

The following definition clarifies what we mean by a model with a unique gapped ground state.
\begin{defn}
We say that a model with an interaction $\Phi\in\caB_f$ has a unique gapped ground state if 
(i)~the $\tau^\Phi$-ground state, which we denote as $\varphi$, is unique, and 
(ii)~there exists a $g>0$ such that
$\sigma(H_{\varphi,\Phi})\setminus\{0\}\subset [g,\infty)$, where  $\sigma(H_{\varphi,\Phi})$ is the spectrum of $H_{\varphi,\Phi}$.
\end{defn}
Note that the uniqueness of $\varphi$ implies that 0 is a non-degenerate eigenvalue of $H_{\varphi,\Phi}$.

If $\varphi$ is a \mbox{$\tau^{\Phi}$-ground} state of reflection invariant interaction $\Phi\in {\caB}_f$,
then its reflection
$ \varphi\circ\gamma$ is also a \mbox{$\tau^{\Phi}$-ground} state.
In particular, if $\varphi$ is a unique \mbox{$\tau^{\Phi}$-ground} state, it is pure and reflection invariant.

In \cite{Matsui2}, T.Matsui showed that the spectral gap implies the split property.
\begin{thm}[Theorem 1.5, Lemma 4.1, and Proposition 4.2 of \cite{Matsui2}]\label{matsui}
Let $\varphi$ be a pure $\tau^\Phi$-ground state of $\Phi\in{\caB}_f$, and denote by  $H_{\varphi,\Phi}$ the corresponding bulk Hamiltonian. 
Assume that $0$ is a non-degenerate eigenvalue of $H_{\varphi,\Phi}$ and
there exists $g>0$ such that $\sigma(H_{\varphi,\Phi})\setminus\{0\}\subset [g,\infty)$.
Then $\varphi$ satisfies the  split property with respect to $\caA_L$ and $\caA_R$.
\end{thm}
%
%\begin{thm}[\cite{Matsui2}Corollary 3.2]\label{matsui}
%Suppose that $\varphi$ is a pure $\tau^\Phi$-ground state of a $\Phi\in\caB_f$ 
%and suppose that its associated bulk Hamiltonian $H_{\varphi,\Phi}$ has a spectral gap.
%%Let $(\caH_{\varphi_R},\pi_{\varphi_R},\Omega_{\varphi_R})$ be the GNS triple of $\varphi_R$, the restriction of $\varphi$ to $\caA_R$.
%Then $\varphi$ satisfies the split property.
%%the von Neumann algebra $\pi_{\varphi_R}\lmk \caA_R\rmk{''}$ is a type I factor.
%\end{thm}
This theorem, combined with Definition \ref{index} allows us to define the $\bbZ_2$-index for
%Hence we can define the $\bbZ_2$-index for 
reflection invariant Hamiltonians with unique gapped ground state.
\begin{defn}
Let $\Phi\in\caB_f$ be a reflection invariant interaction
which has a unique gapped ground state $\omega$.
%Let $\varphi$ be a  a time reversal invariant interaction $\Phi\in \caB_f$ 
%and suppose that its associated bulk Hamiltonian $H_{\varphi,\Phi}$ has a spectral gap.
%Let $(\caH_{\varphi_R},\pi_{\varphi_R},\Omega_{\varphi_R})$ be the GNS triple of $\varphi_R$, the restriction of $\varphi$ to $\caA_R$.
%As $\varphi$ is the unique ground state, it is time reversal invariant, and pure.
%For a time reversal invariant interaction $\Phi\in\caB_f$ with a $\tau^{\Phi}$-unique ground state,
By Theorem \ref{matsui}, $\omega$ satisfies the split property. Hence we obtain the $\bbZ_2$-index $\sigma_\omega$
in Definition  \ref{index}.
In this setting, we denote this $\sigma_\omega$
 by $\hat \sigma_{\Phi}$ and call it the $\bbZ_2$-index associated to $\Phi$.
\end{defn}

%By Matsui \cite{matsui1}, a unique $\tau^\Phi$-ground state of 
%Suppose that $\varphi$ is a pure $\tau^\Phi$-ground state of a $\Phi\in\caB_f$ 
%and suppose that its associated bulk Hamiltonian $H_{\varphi,\Phi}$ has a spectral gap.
%%Let $(\caH_{\varphi_R},\pi_{\varphi_R},\Omega_{\varphi_R})$ be the GNS triple of $\varphi_R$, the restriction of $\varphi$ to $\caA_R$.
%Then $\varphi$ satisfies the split property.
%\begin{prop}\label{index}
%Let $\varphi$ be a unique $\tau^\Phi$-ground state of a time reversal invariant interaction $\Phi\in \caB_f$ 
%and suppose that its associated bulk Hamiltonian $H_{\varphi,\Phi}$ has a spectral gap.
%Let $(\caH_{\varphi_R},\pi_{\varphi_R},\Omega_{\varphi_R})$ be the GNS triple of $\varphi_R$, the restriction of $\varphi$ to $\caA_R$.
%Then the conclusion of Proposition \ref{index} holds for 
% $\varphi$.
%\end{prop}

Since $\hat \sigma_\Phi$ takes discrete values $\{-1,1\}$,
for a continuous path of interactions $\Phi(s)$, we would expect that
$\hat \sigma_{\Phi(s)}$ is constant.
We prove this in the setting of $C^1$-classification.
\begin{defn}\label{boundary}
We say the map $\Phi:[0,1]\ni s \to \Phi(s):=\{\Phi(X;s)\}_{X\in {\mathfrak S}_\bbZ}\in {\caB_f}$ is a $C^1$-path of reflection invariant gapped interactions satisfying the {\it Condition B},
if 
there exist 
\begin{description}
\item[(i)] numbers
$M,R\in\nan$, $g>0$  and an increasing sequence $n_k\in\nan$, $k=1,2,\ldots$,
\item[(ii)] $C^1$-functions
$a,b:[0,1]\to \bbR$ such that $a(s)<b(s)$,
\item[(iii)]
a sequence of paths of interactions $\Psi_k:[0,1]\ni s \to \Psi_k(s):=\{\Psi_k(X;s)\}_{X\in {\mathfrak S}_\bbZ}\in {\caB_f}$, $k=1,2,\ldots$,
\end{description}
and the following hold:
\begin{enumerate}
\item For each $X\in{\mathfrak S}_\bbZ$, the map
$[0,1]\ni s\to \Phi(X;s), \Psi_k(X;s)\in\caA_{X}$ are continuous and piecewise $C^1$.
We denote by $\Phi'(X;s)$, $\Psi'_k(X;s)$, the corresponding derivatives.
\item For each $s\in[0,1]$, and $X\in{\mathfrak S}_\bbZ$ with $\diam (X)\ge M$, we have
$\Phi(X;s)=0$.
\item For each $s\in[0,1]$, and $k\in\bbN$, we have $\Psi_k(X;s)=0$ unless $X\subset \Lambda_{n_k}\setminus \Lambda_{n_k-R}$.
\item Interactions are bounded as follows
\begin{align}
C_1:=\sup_{s\in[0,1]}\sup_{k\in\nan}\sup_{X\in {\mathfrak S}_\bbZ}
\lmk
\lV
\Phi\lmk X;s\rmk
\rV+|X|\lV
\Phi' \lmk X;s\rmk
\rV
+
\lV
\Psi_k\lmk X;s\rmk
\rV+|X|\lV
\Psi_k'\lmk X;s\rmk
\rV
\rmk<\infty.
\end{align}
\item For each $s\in[0,1]$, there exists a unique $\tau^{\Phi(s)}$-ground state
$\varphi_s$. 
\item
For each $s\in[0,1]$, $\Phi(s)$ is reflection-invariant.
\item  
 For each $k\in\bbN$ and $s\in[0,1]$, 
 the spectrum $\sigma\lmk \lmk H_{\Phi(s)+\Psi_k(s)}\rmk_{\Lambda_{n_k}}\rmk$
of $ \lmk H_{\Phi(s))+\Psi_k(s)}\rmk_{\Lambda_{n_k}}$ is decomposed into two non-empty disjoint parts
$
\sigma\lmk \lmk H_{\Phi(s)+\Psi_k(s)}\rmk_{\Lambda_{n_k}}\rmk=
\Sigma_1^{(k)}(s)\cup\Sigma_2^{(k)}(s)
$
such that $\Sigma_1^{(k)}(s)\subset [a(s),b(s)]$,
$\Sigma_2^{(k)}(s)\subset [b(s)+g,\infty)$ and 
the diameter of $\Sigma_1^{(k)}(s)$ converges to $0$ as $k\to\infty$.
\end{enumerate}
\end{defn}
The interaction $\Psi_k(s)$ corresponds to a boundary condition. Note that it does not forbid an interaction
between intervals $[-n,-n+R]\cap \bbZ$ and $[n-R,n]\cap \bbZ$.
In particular, the periodic boundary condition is included in this framework.
Also, note that we {\it do not} require that the boundary term $\Psi_k(s)$
to be reflection invariant. 

By exactly the same way as in Proposition 3.5 of \cite{Ogata4}, we can show the following.:
\begin{prop}\label{aep}
Let $\Phi:[0,1]\ni s \to \Phi(s):=\{\Phi(X;s)\}_{X\in {\mathfrak S}_\bbZ}\in {\caB_f}$ be a $C^1$-path of reflection invariant gapped interactions satisfying the {\it Condition B}.
Let $\varphi_s$ be the unique $\tau^{\Phi(s)}$-ground state, for each $s\in[0,1]$. 
Then
$\varphi_0$ and $\varphi_1$
are automorphic equivalent via an automorphism, 
which allows a reflection invariant decomposition.
\end{prop}
%Note that {\it Condition B} implies that for each $s\in[0,1]$, $\Phi(s)$ has a unique gapped ground state.
As a corollary of this proposition and Theorem \ref{aeii}, we obtain the following.
\begin{thm}\label{c1t}
Let $\Phi:[0,1]\ni s \to \Phi(s):=\{\Phi(X;s)\}_{X\in {\mathfrak S}_\bbZ}\in {\caB_f}$ be a $C^1$-path of reflection invariant gapped interactions satisfying the {\it Condition B}.
Then
we have
$\hat \sigma_{\Phi(0)}=\hat \sigma_{\Phi(1)}$.
\end{thm}
Namely, the $\bbZ_2$-index is invariant along the $C^1$-path of reflection invariant gapped interactions, satisfying the {\it Condition B}.

\section{The $\bbZ_2$-index and the modular conjugation}
In this section, we give a characterization of $\sigma_\omega$ from the point of view of 
Tomita-Takesaki modular theory.
It will be used in Section \ref{poco}, to prove that our index generalizes the index introduced in \cite{po}.
This also allows us to connect $\sigma_\omega$ with the Schmidt decomposition of $\Omega_\omega$.

First let us recall Tomita-Takesaki theory. See \cite{BR1} or \cite{takesaki2} for more information.
Let $\caM$ be a von Neumann algebra acting on a Hilbert space $\caH$.
Let $\Omega$ be a cyclic (i.e. $\caM\Omega$ is dense in $\caH$) and separating
(i.e., $x\Omega=0$, $x\in\caM$ implies $x=0$) vector
for $\caM$.
We define an anti-linear operator on $\caH$ with domain
$\caM\Omega$ by
\begin{align}\label{sdef}
Sx\Omega:=x^*\Omega,\quad x\in\caM.
\end{align}
It turns out that $S$ is closable. (Proposition 2.5.9 of \cite{BR1}.) We denote the closure by the same symbol $S$.
The operator $S$ has a polar decomposition $S=J\Delta^{\frac12}$ where $J$ is an anti-unitary $J$ called the modular conjugation associated to $(\caM,\Omega)$ and $\Delta$ is a nonsingular positive operator called modular operator associated to $(\caM,\Omega)$.
For the commutant  $\caM'$ of $\caM$, we have 
\begin{align}\label{fdef}
J\Delta^{-\frac 12}x'\Omega
=\lmk x'\rmk^*\Omega,\quad x'\in\caM'.
\end{align}
We also have
\begin{align}\label{add}
\Delta^{-\frac 12}=J\Delta^{\frac 12} J^*,\quad J^2=\unit,\quad
\Delta\Omega=J\Omega=\Omega.
\end{align}
(Proposition 2.5.11 of \cite{BR1}.)
The subspace $\caM'\Omega$ is a core of $J\Delta^{-\frac 12}$.
The Tomita-Takesaki theory states that
\begin{align}
\Delta^{it}\caM\Delta^{-it}=\caM, \; \text{for all } \; t\in\bbR\quad\text{and}\quad
J\caM J^*=\caM'.
\end{align}
From the first property, we may define a $W^*$-dynamics (i.e.,
$\sigma$-weak continuous one parameter group of automorphisms)
$\sigma_t(x):=\Delta^{it}x\Delta^{-it}$, $t\in\bbR$, $x\in\caM$ on
$\caM$. It is called the modular automorphisms associated to $(\caM,\Omega)$.

Set $\caD:=\{z\in\bbC\mid 0< \Im z< 1\}$ and denote by $\bar\caD$ its closure.
For any $x,y\in\caM$, there exists a bounded and continuous 
function $F_{x,y}$ on $\bar\caD$ which is analytic on $\caD$,
satisfying
\begin{align}\label{fxy}
F_{x,y}(t)=\braket{\Omega}{\sigma_t(x)y\Omega},\quad
F_{x,y}(t+i)=\braket{\Omega}{y\sigma_t(x)\Omega},
\end{align}
 for all $t\in\bbR$.
 This condition is called the KMS-condition for the positive linear functional $\caM\ni x\mapsto \braket{\Omega}{x \Omega}$ on $\caM$
 and 
the modular automorphisms are characterized as the unique $W^*$-dynamics
which satisfies the KMS condition for this linear functional. (Theorem 1.2 VIII \cite{takesaki2}.)

Now let us come back to our problem.
\begin{lem}\label{bas}
Let $\omega$ be
a reflection invariant pure state on $\caA$ which satisfies the split property with respect to $\caA_R$ and
$\caA_L$. 
Let $(\caH_\omega, \pi_\omega, \Omega_\omega, \Gamma_\omega)$  be a reflection-spilt representation associated to
$\omega$.
Let $s_\omega$ be a projection in $B(\caH_\omega)$ such that 
the support projection of $\Omega_\omega$ in 
$\unit_{\caH_\omega}\otimes B(\caH_\omega)$ is $\unit_{\caH_\omega}\otimes s_\omega$.
Set $\caM:=s_\omega\otimes B(s_\omega\caH_\omega)$, and
$p_\omega=s_\omega\otimes s_\omega$.
Then 
\begin{enumerate}
\item the support projection of $\Omega_\omega$ in $ B(\caH_\omega)\otimes \unit_{\caH_\omega}$ is $s_\omega\otimes \unit_{\caH_\omega}$,
\item $p_\omega\Omega_\omega=\Omega_\omega$,
\item the commutant $\caM'$ of $\caM$ in $p_\omega\hat \caH_\omega$  is 
$\caM'= B(s_\omega\caH_\omega)\otimes s_\omega$,
\item $\Omega_\omega$ is cyclic and separating for $\caM$ in $p_\omega\hat\caH_\omega$,
%\item $s_\omega\otimes \unit_{\caH_\omega}$ is the support projection of
%$\Omega_\omega$ in $B(\caH_\omega)\otimes \unit_{\caH_\omega}$,
\item $\Gamma_\omega p_\omega=p_\omega \Gamma_\omega$.
\end{enumerate}
\end{lem}
\begin{proof}
To show {\it 1.},
let $s'_\omega$ be a projection in $B(\caH_\omega)$ such that 
the support projection of $\Omega_\omega$ in 
$B(\caH_\omega)\otimes \unit_{\caH_\omega}$ is $s'_\omega\otimes \unit_{\caH_\omega}$.
From (\ref{gxg}) and (\ref{guni}) with $A=\unit_{\caA}$, we have
\begin{align}
\braket{\Omega_\omega}{\lmk (1-s_\omega )\otimes\unit_{\caH_\omega} 
\rmk \Omega_\omega}
=
\braket{\Omega_\omega}{\Gamma_\omega \lmk(1- s_\omega)\otimes \unit_{\caH_\omega} 
\rmk \Gamma_\omega^*\Omega_\omega}
=\braket{\Omega_\omega}{\lmk \unit_{\caH_\omega} \otimes (1-s_\omega)
\rmk \Omega_\omega}=0.
\end{align}
Therefore, we have $1-s_\omega\le 1-s'_\omega$. Similarly, we obtain
$1-s'_\omega\le 1-s_\omega$. Hence we obtain $s_\omega=s'_\omega$.
{\it 2.} is clear from the definition.
{\it 3.} 
follows from Tomita's commutant Theorem (See Theorem 5.9 of \cite{takesaki}).
Since $\Omega_\omega$ is separating for $\caM'$, it is cyclic for
$\caM$ in $p_\omega\hat \caH_\omega$, {\it 4}.
(Proposition 2.5.3 \cite{BR1}.)
{\it 5.} is from the definition of $p_\omega=s_\omega\otimes s_\omega$ and (\ref{gxg}).
\end{proof}

From {\it 4.} of Lemma \ref{bas},  we can define modular conjugation $J_\omega$ and
modular operator $\Delta_\omega$ associated to
$(\caM,\Omega_\omega)$.
Let us investigate their properties.
\begin{lem}\label{thir}
Let $\omega$ be
a reflection invariant pure state on $\caA$ which satisfies the split property with respect to $\caA_R$ and
$\caA_L$. 
Let $(\caH_\omega, \pi_\omega, \Omega_\omega, \Gamma_\omega)$  be a reflection-spilt representation associated to
$\omega$.
Let $s_\omega$ be a projection in $B(\caH_\omega)$ such that 
the support projection of $\Omega_\omega$ in 
$\unit_{\caH_\omega}\otimes B(\caH_\omega)$ is $\unit_{\caH_\omega}\otimes s_\omega$.
Set $\caM:=s_\omega\otimes B(s_\omega\caH_\omega)$, and
$p_\omega=s_\omega\otimes s_\omega$.
Let $J_\omega$, $\Delta_\omega$ be modular conjugation, modular operator on 
$p_\omega\hat\caH_\omega$ associated to
$(\caM,\Omega_\omega)$.
Let
\begin{align}
\Omega_\omega=\sum_{k\in \Lambda} \sqrt{\lambda_k} 
\xi_k\otimes \zeta_k
\end{align}
be a Schmidt decomposition of $\Omega_\omega$.
Here $\Lambda$ is a countable set and the sequence $\{\lambda_k\}_{k\in\Lambda}\subset\bbR_{>0}$
satisfies $\sum_{k\in\Lambda}\lambda_k=1$. 
Furthermore, each of $\{\xi_k\}_{k\in\Lambda}$ and $\{\zeta_k\}_{k\in\Lambda}$
are orthonormal sets of $\caH_\omega$.
We also define a density matrix $\rho_\omega$ by
\begin{align}
\rho_\omega=\sum_{k\in\Lambda} \lambda_k \ketbra{\zeta_k}{\zeta_k}.
\end{align}
Then we have the following.:
\begin{enumerate}
\item
Both of $\{\xi_k\}_{k\in\Lambda}$ and $\{\zeta_k\}_{k\in\Lambda}$ are orthonormal basis of $s_\omega\caH_\omega$.
There exists a unitary $u$ on $s_\omega\caH_\omega$ such that
$\xi_k=u\zeta_k$, for each $k\in\Lambda$. 
\item 
The action of the modular operator is given by
\begin{align}
\Delta_\omega^{\frac 12}\lmk s_\omega\otimes x\rmk\Omega_\omega
=\sum_{k\in\Lambda} u\zeta_k\otimes \rho_\omega^{\frac 12} x\zeta_k,\quad
x\in B(s_\omega\caH_\omega).
\end{align}
In particular, if the rank of $\rho_\omega$
is finite  then we have
\begin{align}
\Delta_\omega^{\frac 12}\lmk s_\omega\otimes x\rmk\Omega_\omega
=\lmk s_\omega\otimes \rho_\omega^{\frac 12}x\rho_\omega^{-\frac 12}\rmk\Omega_\omega,
\quad
x\in B(s_\omega\caH_\omega).
\end{align}
\item
Let $c$ be the complex conjugation on $s_\omega \caH_\omega$ given by
$c\zeta_k=\zeta_k$ for $k\in\Lambda$.
Then we have
\begin{align}
\lmk s_\omega\otimes x\rmk\Omega_\omega
=\sum_{l\in\Lambda} u\rho_\omega^{\frac 12} c^* x^*c\zeta_l\otimes \zeta_l,\quad
x\in B(s_\omega\caH_\omega).
\end{align}
\item
The adjoint of the modular conjugation on $\caM$
is given by
\begin{align}
J_\omega\lmk
s_\omega\otimes x
\rmk
J_\omega^*
=uc^* xc u^*\otimes s_\omega,\quad x\in B(s_\omega\caH_\omega).
\end{align}

\end{enumerate}
\end{lem}
\begin{proof}

By the definition of $s_\omega$ and {\it 1.} of Lemma \ref{bas}, we see that
each of $\{\xi_k\}_{k\in\Lambda}$ and $\{\zeta_k\}_{k\in\Lambda}$
are orthonormal basis of $s_\omega \caH_\omega$.
Therefore, there is a unitary $u$ on $s_\omega \caH_\omega$
such that $\xi_k=u\zeta_k$, for all $k\in\Lambda$.
This $u$ is given as
$u=\sum_{k\in\Lambda}\ketbra{\xi_k}{\zeta_k}$
where the summation converges in the strong topology.
This proves {\it 1.}

Let $\Lambda_n$, $n\in\nan$ be an increasing sequence of finite subsets of $\Lambda$ such that
$\Lambda_n\nearrow \Lambda$. 
Set 
\begin{align}
Q_n:=\sum_{k\in\Lambda_n} \ketbra{\zeta_k}{\zeta_k}.
\end{align}
For $x\in B(s_\omega\caH_\omega)$ and $n\in\nan$, set $x_n:=Q_n x Q_n$.
Note that the sequence $x_n\in B(s_\omega\caH_\omega)$ approximates $x$ in the $\sigma$-strong$^*$ topology.

We would like to specify the action of the modular operator {\it 2.} We claim
\begin{align}\label{ikms}
\Delta_\omega^{\frac 12}\lmk s_\omega\otimes x\rmk\Omega_\omega
=\sum_{k\in\Lambda} u\zeta_k\otimes \rho_\omega^{\frac 12} x\zeta_k,\quad
x\in B(s_\omega\caH_\omega).
\end{align}
Note that the right hand side converges in norm because $\rho_\omega$ is in the trace class.
To prove (\ref{ikms}), we first specify the modular automorphism $\sigma$ with respect to
$(\caM,\Omega_\omega)$.
We define a $W^*$-dynamics $\alpha$ on $\caM$ by
\begin{align}
\alpha_t(s_\omega\otimes x)
:=s_\omega\otimes \Ad(\rho_\omega^{it})(x),\quad t\in\bbR,\; x\in B(s_\omega\caH_\omega),
\end{align}
and show that $\alpha=\sigma$.
To do that, we recall the uniqueness of the $W^*$-dynamics which satisfies the KMS condition.
Let $\varphi$ be a state on $\caM$ given by
\begin{align}
\varphi\lmk s_\omega\otimes x\rmk
=\braket{\Omega_\omega}{\lmk
s_\omega\otimes x
\rmk
\Omega_\omega},
\quad x\in B(s_\omega \caH_\omega).
\end{align}
Note that 
\begin{align}\label{rp}
\Tr_{\caH_\omega}\lmk\rho_\omega x\rmk=
\braket{\Omega_\omega}{\lmk
s_\omega\otimes x\rmk
\Omega_\omega}=\varphi\lmk s_\omega\otimes x\rmk,\quad
x\in B(s_\omega\caH_\omega).
\end{align}
From this, we can see that $\varphi$ is $\alpha$-invariant.
We show that $\alpha$ satisfies the KMS-condition for $\varphi$.
This follows from the standard argument like in Proposition 5.3.7 of \cite{BR2}.: 
For any $x,y\in B(s_\omega\caH_\omega)$, and $n\in\nan$, 
setting $x_n:=Q_n x Q_n$,
we may define an entire analytic function
\begin{align}
F_{x_n,y}(z):=\braket{\Omega_\omega}{ \alpha_z\lmk s_\omega\otimes x_n\rmk
\lmk s_\omega\otimes y\rmk\Omega_\omega},\quad z\in\bbC,
\end{align}
because 
$\bbR\ni t\mapsto \alpha_t(s_\omega\otimes x_n)\in \caM$ has an analytic continuation
$\alpha_z(s_\omega\otimes x_n)=s_\omega\otimes \Ad(\rho_\omega^{iz})(x_n)\in\caM$, $z\in\bbC$.
This $F_{x_n,y}(z)$  is bounded and continuous on $\bar\caD$ and
analytic on $\caD$.
Furthermore, it satisfies the boundary condition (\ref{fxy}) i.e.,
we have
\begin{align}\label{bound}
F_{x_n,y}(t)=\braket{\Omega_\omega}{\alpha_t\lmk s_\omega\otimes x_n\rmk \lmk
s_\omega\otimes y\rmk
\Omega_\omega},\quad
F_{x_n,y}(t+i)=\braket{\Omega_\omega}{ \lmk s_\omega\otimes y\rmk
\alpha_t\lmk  s_\omega\otimes x_n)\rmk\Omega_\omega},
\end{align}
for all $t\in\bbR$.
The second property holds because of (\ref{rp}) and the property of the trace.

By the $\alpha$-invariance of $\varphi$ and the $\sigma$-strong$^*$ convergence of
$x_n$ to $x$, (using the Cauchy-Schwartz inequality,)
one can show from (\ref{bound}) that
$F_{x_n,y}(t)$ and $F_{x_n,y}(t+i)$, as functions of $t\in\bbR$,
are Cauchy sequence of continuous bounded functions on $\bbR$
with respect to the uniform norm.
Therefore, 
by the Phragmen-Lindel\"{o}f theorem, $\{F_{x_n,y}(z)\}_n$  is a Cauchy sequence of continuous bounded functions on $\bar\caD$ with respect to the uniform norm.
Therefore, $F_{x_n,y}(z)$ has a limit $F_{x,y}(z)$ on $z\in \bar \caD$, which 
is bounded and continuous on $\bar \caD$ and analytic on $\caD$.
We also have
\begin{align}
&F_{x,y}(t)=\lim_{n\to\infty}F_{x_n,y}(t)=\braket{\Omega_\omega}{ \alpha_{t}(s_\omega\otimes x)\lmk s_\omega\otimes y\rmk \Omega_\omega},\\
&F_{x,y}(t+i)=\lim_{n\to\infty}F_{x_n,y}(t+i)=\braket{\Omega_\omega}{ \lmk
s_\omega\otimes y\rmk \alpha_{t}(s_\omega\otimes x)\Omega_\omega},
\end{align}
for any $t\in\bbR$.
Therefore, $\alpha$ satisfies the KMS condition for $\varphi$, and from the uniqueness, we get 
$\alpha=\sigma$.
Now let us prove (\ref{ikms}).
For each $x\in B(s_\omega\caH_\omega)$, we again consider $x_n=Q_nx Q_n$, $n\in\nan$.
By the entire analyticity of $s_\omega\otimes x_n$ with respect to $\alpha=\sigma$,
we have
\begin{align}\label{61}
\Delta_\omega^{\frac 12} \lmk s_\omega\otimes x_n\rmk\Omega_\omega
=\alpha_{ -\frac i 2} \lmk s_\omega\otimes x_n\rmk \Omega_\omega
=\lmk s_\omega\otimes \Ad(\rho_\omega^{\frac {1}2})(x_n)\rmk\Omega_\omega
=\sum_{k\in\Lambda} u\zeta_k\otimes \rho_\omega^{\frac 12} x_n\zeta_k.
\end{align}
(The first equation is a standard argument. See proof of Theorem 5.5 of \cite{DJP} for example.)
The left hand side of (\ref{61}) converges to $\Delta_\omega^{\frac 12} \lmk s_\omega\otimes x\rmk\Omega_\omega$
 as $n\to\infty$ because of 
 \begin{align}
 \lV
 \Delta_\omega^{\frac 12} \lmk s_\omega\otimes x_n\rmk\Omega_\omega
 -\Delta_\omega^{\frac 12} \lmk s_\omega\otimes x\rmk\Omega_\omega
 \rV
 =\lV
 J_\omega\Delta_\omega^\frac 12\lmk s_\omega\otimes x_n\rmk\Omega_\omega
 -J_\omega\Delta_\omega^\frac 12\lmk s_\omega\otimes x\rmk\Omega_\omega
 \rV
 =\lV
 \lmk s_\omega\otimes (x_n^*-x^*)\rmk\Omega_\omega
 \rV,
 \end{align}
and the $\sigma$-strong$^*$-convergence $x_n\to x$.
The right hand of (\ref{61}) converges to $\sum_{k\in\Lambda} u\zeta_k\otimes \rho_\omega^{\frac 12} x\zeta_k$
because of
\begin{align}
\lV
\sum_{k\in\Lambda} u\zeta_k\otimes \rho_\omega^{\frac 12} (x_n-x)\zeta_k
\rV^2
=\sum_{k\in\Lambda}  \lV \rho_\omega^{\frac 12} (x_n-x)\zeta_k\rV^2,
\end{align}
and the $\sigma$-strong$^*$-convergence $x_n\to x$.
Hence we have proven (\ref{ikms}).

Next we show {\it 3.}
\begin{align}\label{6ii}
\lmk s_\omega\otimes x\rmk\Omega_\omega
=\sum_{l\in\Lambda} u\rho_\omega^{\frac 12} c^* x^*c\zeta_l\otimes \zeta_l,\quad
x\in B(s_\omega\caH_\omega).
\end{align}
The right hand side converges in norm.
To prove (\ref{6ii}), first we consider
\begin{align}\label{tot2}
\lmk
s_\omega\otimes Q_n xQ_m
\rmk\Omega_\omega
=\sum_{k\in\Lambda_m}\sum_{l\in\Lambda_n} 
\sqrt{\lambda_k}u\zeta_k\otimes
\ket{\zeta_l}\braket{\zeta_l}{x\zeta_k},
\end{align}
for $n,m\in\nan$.
Since we have $\braket{\zeta_l}{x\zeta_k}=\braket{c\zeta_l}{xc\zeta_k}=\braket{c^* xc \zeta_k}{\zeta_l}$ 
by $c\zeta_k=\zeta_k$, 
we have
\begin{align}
(\ref{tot2})=
\sum_{l\in\Lambda_n} \sum_{k\in\Lambda_m}
\sqrt{\lambda_k}\ket{u\zeta_k}\braket{c^* xc \zeta_k}{\zeta_l}\otimes\ket{\zeta_l}
=\sum_{l\in\Lambda_n} u Q_m \rho_\omega^{\frac 12} Q_m c^* x^*c\zeta_l\otimes \zeta_l.
\end{align}
Hence we obtain
\begin{align}\label{ext}
\lmk
s_\omega\otimes Q_n xQ_m
\rmk\Omega_\omega
=\sum_{l\in\Lambda_n} u Q_m \rho_\omega^{\frac 12} Q_m c^* x^*c\zeta_l\otimes \zeta_l.
\end{align}
Taking $m\to\infty$, and then $n\to\infty$,
we obtain (\ref{6ii}).
%\begin{align}
%\lmk
%s_\omega\otimes x
%\rmk\Omega_\omega
%=\sum_{l\in\Lambda} u  \rho_\omega^{\frac 12} c^* x^*c\zeta_l\otimes \zeta_l.
%\end{align}

Next we consider the action of $J_\omega$ {\it 4}. 
We claim
\begin{align}\label{lml}
J_\omega\lmk
s_\omega\otimes x
\rmk
J_\omega^*
=uc^* xc u^*\otimes s_\omega,\quad x\in B(s_\omega\caH_\omega).
\end{align}
To prove this, note that
\begin{align}
J_\omega\lmk
s_\omega\otimes x
\rmk\Omega_\omega
=J_\omega J_\omega\Delta_\omega^\frac 12\lmk
s_\omega\otimes x^*
\rmk\Omega_\omega
=\Delta_\omega^{\frac 12}
\lmk
s_\omega\otimes x^*
\rmk\Omega_\omega
=\sum_{k\in\Lambda}u\zeta_k\otimes \rho_\omega^{\frac 12} x^*\zeta_k
=\lim_{n\to\infty} \sum_{k\in\Lambda_n}u\zeta_k\otimes \rho_\omega^{\frac 12} x^*\zeta_k,
\end{align}
for any $ x\in B(s_\omega\caH_\omega)$, by {\it 2}.
The right hand side converges in norm.
Therefore, for any $m\in\nan$, we have
\begin{align}
\lmk
s_\omega\otimes Q_m
\rmk
J_\omega\lmk
s_\omega\otimes x
\rmk\Omega_\omega
=\lim_{n\to\infty} 
\sum_{k\in\Lambda_n}u\zeta_k\otimes Q_m\rho_\omega^{\frac 12} x^*\zeta_k,\quad
 x\in B(s_\omega\caH_\omega).
\end{align}
Note that as in the proof of (\ref{ext}), we have
\begin{align}
&
\sum_{k\in\Lambda_n}u\zeta_k\otimes Q_m\rho_\omega^{\frac 12} x^*\zeta_k
=
\sum_{l\in \Lambda_m}\sum_{k\in\Lambda_n}u\zeta_k\otimes \ket{\zeta_l}\braket{\zeta_l}{\rho_\omega^{\frac 12} x^*\zeta_k}\nonumber\\
&=\sum_{l\in \Lambda_m}\sum_{k\in\Lambda_n}u\zeta_k\braket{c^*\rho_\omega^{\frac 12} x^*c\zeta_k}{\zeta_l}\otimes {\zeta_l}=\sum_{l\in \Lambda_m}\sum_{k\in\Lambda_n}u
\ketbra{\zeta_k}{\zeta_k} c^*x\rho_\omega^{\frac 12} c\zeta_l\otimes\zeta_l
=\lmk uQ_n u^*\otimes s_\omega\rmk
 \sum_{l\in \Lambda_m}uc^* x\rho_\omega^{\frac 12}c\zeta_l\otimes \zeta_l\nonumber\\
&=\lmk uQ_n u^*\otimes s_\omega\rmk
\sum_{l\in \Lambda_m}\sqrt{\lambda_l}uc^* xc u^* u \zeta_l\otimes \zeta_l
=\lmk uQ_n u^*uc^* xc u^*\otimes s_\omega\rmk
\sum_{l\in \Lambda_m}\sqrt{\lambda_l} u \zeta_l\otimes \zeta_l
=\lmk uQ_n u^*uc^* xc u^*\otimes s_\omega\rmk(s_\omega\otimes Q_m)\Omega_\omega,
\end{align}
for any $ x\in B(s_\omega\caH_\omega)$.
Taking $n\to\infty$, we obtain
\begin{align}
\lmk
s_\omega\otimes Q_m
\rmk
J_\omega\lmk
s_\omega\otimes x
\rmk\Omega_\omega
=\lim_{n\to\infty} 
\sum_{k\in\Lambda_n}u\zeta_k\otimes Q_m\rho_\omega^{\frac 12} x^*\zeta_k
=\lmk uc^* xc u^*\otimes Q_m\rmk\Omega_\omega,
\end{align}
for any $ x\in B(s_\omega\caH_\omega)$.
Taking $m\to\infty$, we obtain
\begin{align}
J_\omega\lmk
s_\omega\otimes x
\rmk
J_\omega^*
\Omega_\omega
=J_\omega\lmk
s_\omega\otimes x
\rmk\Omega_\omega
=\lmk uc^* xc u^*\otimes s_\omega\rmk\Omega_\omega,\quad x\in B(s_\omega\caH_\omega),
\end{align}
for any $ x\in B(s_\omega\caH_\omega)$.
Note that $J_\omega\lmk
s_\omega\otimes x
\rmk
J_\omega^*\in \Ad(J_\omega)(\caM)=\caM'$ and
$ uc^* xc u^*\otimes s_\omega\in\caM'$.
Since $\Omega_\omega$ is separating for $\caM'$,
we have
\begin{align}
J_\omega\lmk
s_\omega\otimes x
\rmk
J_\omega^*
=uc^* xc u^*\otimes s_\omega,\quad x\in B(s_\omega\caH_\omega),
\end{align}
proving the claim.

\end{proof}

The adjoint action of $J_\omega$ on $\caM$ introduces a $\bbZ_2$-index.
\begin{prop}\label{kdef}
Let $\omega$ be
a reflection invariant pure state on $\caA$ which satisfies the split property with respect to $\caA_R$ and
$\caA_L$. 
Let $(\caH_\omega, \pi_\omega, \Omega_\omega, \Gamma_\omega)$  be a reflection-spilt representation associated to
$\omega$.
Let $s_\omega$ be a projection in $B(\caH_\omega)$ such that 
the support projection of $\Omega_\omega$ in 
$\unit_{\caH_\omega}\otimes B(\caH_\omega)$ is $\unit_{\caH_\omega}\otimes s_\omega$.
Set $\caM:=s_\omega\otimes B(s_\omega\caH_\omega)$, and
$p_\omega=s_\omega\otimes s_\omega$.
Let $J_\omega$ be the modular conjugation on $p_\omega\hat\caH_\omega$ associated to
$(\caM,\Omega_\omega)$.
Then we have the following:
\begin{enumerate}
\item
There exists an anti-unitary $\theta: s_\omega\caH_\omega\to s_\omega\caH_\omega$
such that 
\begin{align}
J_\omega\lmk s_\omega\otimes x\rmk J_\omega^*
=\theta x\theta^*\otimes s_\omega,\label{jtj}
%&J_\omega\lmk x\otimes s_\omega\rmk J_\omega^*
%=s_\omega\otimes \theta x\theta^*,
\end{align}
for all $x\in B(s_\omega\caH_\omega)$.
\item
For any anti-unitary $\theta: s_\omega\caH_\omega\to s_\omega\caH_\omega$ satisfying
(\ref{jtj}), we have
\begin{align}
J_\omega\lmk x\otimes s_\omega\rmk J_\omega^*
=s_\omega\otimes \theta x\theta^*,
\end{align}
for all $x\in B(s_\omega\caH_\omega)$.
\item There exists a $\kappa_\omega\in\{\pm 1\}$ satisfying $\theta^2=\kappa_\omega s_\omega$
for any anti-unitary $\theta: s_\omega\caH_\omega\to s_\omega\caH_\omega$
satisfying (\ref{jtj}).
\end{enumerate}
\end{prop}
\begin{proof}
Let $\Delta_\omega$ be the modular operator of
$(\caM,\Omega_\omega)$.
From {\it 5.} of Lemma \ref{bas}, 
$\tilde\Gamma:=\Gamma_\omega p_\omega $ defines a unitary operator
on $p_\omega\hat\caH_\omega$.
Note that $\tilde \Gamma\Omega_\omega=\Omega_\omega$, because of Lemma \ref{bas} {\it 2.} and Definition \ref{rsr}
{\it 4.}
We claim that 
\begin{align}\label{tgp}
J_\omega\tilde\Gamma=\tilde \Gamma J_\omega.
\end{align}
To see this, we recall (\ref{sdef}) and (\ref{fdef}).
For any $x\in\caM$, we have
\begin{align}\label{tochu}
\tilde\Gamma J_\omega\Delta_\omega^{\frac 12}x\Omega_\omega=\tilde\Gamma x^*\Omega_\omega=\tilde\Gamma x^*\tilde\Gamma^*\tilde\Gamma\Omega_\omega
=\tilde\Gamma x^*\tilde\Gamma^*\Omega_\omega.
\end{align}
Note that from (\ref{gxg}), the element $\tilde\Gamma x^*\tilde\Gamma^*$ belongs to $\caM'$.
Therefore, from (\ref{fdef}), we have
\begin{align}
(\ref{tochu})=J_\omega\Delta_\omega^{-\frac 12}\tilde\Gamma x\tilde\Gamma^*\Omega_\omega
=J_\omega\Delta_\omega^{-\frac 12}\tilde\Gamma x\Omega_\omega.
\end{align}
Since $\caM\Omega_\omega$ is a core of $J_\omega\Delta_\omega^{\frac12}$ and $\caM'\Omega_\omega$
 is a core of $J_\omega\Delta_\omega^{-\frac12}$,
 this means 
 \begin{align}
 \tilde\Gamma J_\omega\Delta_\omega^{\frac 12}
 =J_\omega\Delta_\omega^{-\frac 12}\tilde\Gamma
 =J_\omega\tilde\Gamma\tilde\Gamma^*\Delta_\omega^{-\frac 12}\tilde\Gamma.
 \end{align}
 By the uniqueness of the polar decomposition, we obtain 
 $\tilde\Gamma J_\omega=J_\omega\tilde\Gamma$, proving the claim.

Next we note that there are anti-$*$--automorphisms $\Theta_{L\to R},\Theta_{R\to L}$
on $B(s_\omega\caH_\omega)$ such that 
\begin{align}
&J_\omega\lmk s_\omega\otimes x\rmk J_\omega^*
=\Theta_{R\to L}(x)\otimes s_\omega,\label{jj1}\\
&J_\omega\lmk x\otimes s_\omega\rmk J_\omega^*
=s_\omega\otimes \Theta_{L\to R}(x)\label{jj2},
\end{align}
for any $x\in B(s_\omega\caH_\omega)$.
This is because of the Tomita-Takesaki theory, i.e.,
$J_\omega\caM J_\omega^*=\caM'$ and 
$J_\omega\caM' J_\omega^*=\caM$.
Note that $\Theta_{L\to R}\circ \Theta_{R\to L}=\id=\Theta_{R\to L}\circ\Theta_{L\to R}$, because of
$J_\omega^2=s_\omega$.
By (\ref{gxg}) and (\ref{tgp}), this $\Theta_{L\to R}$ and $\Theta_{R\to L}$
coinsides.:For any $x\in B(s_\omega\caH_\omega)$,
we have 
\begin{align}
&\Theta_{R\to L}(x)\otimes s_\omega
=(\Ad J_\omega)\lmk s_\omega\otimes x\rmk
=\lmk \Ad J_\omega\tilde\Gamma\rmk
\lmk x\otimes s_\omega\rmk
=\lmk \Ad \tilde\Gamma J_\omega\rmk
\lmk x\otimes s_\omega\rmk\nonumber\\
&=\Ad \tilde\Gamma \lmk
s_\omega\otimes \Theta_{L\to R}(x)
\rmk
=\Theta_{L\to R}(x)\otimes 
s_\omega.
\end{align}
Hence we have $\Theta_{R\to L}(x)=\Theta_{L\to R}(x)$ for any $x\in B(s_\omega\caH_\omega)$.

Now let us prove {\it 1.-3.} of the Lemma.
{\it 1}. is shown in Lemma \ref{thir} {\it 4}, as $\theta=uc^*$.
%By Wigner's theorem 
%there exists an anti-unitary $\theta: s_\omega\caH_\omega\to s_\omega\caH_\omega$
%such that 
%\begin{align}
%\Theta_{R\to L}(x)=\Theta_{L\to R}(x)=\theta x\theta^*,
%\end{align}
%for any $x\in B(s_\omega\caH_\omega)$.
%Hence from (\ref{jj1}) and ({\ref{jj2}}), we obtain {\it 1.} of the Proposition.
%
To prove the second and the third statement, let
${\theta}: s_\omega\caH_\omega\to s_\omega\caH_\omega$ be any anti-unitary 
such that 
\begin{align}\label{jtjj}
J_\omega\lmk s_\omega\otimes x\rmk J_\omega^*
={\theta} x{\theta}^*\otimes s_\omega
,\quad x\in B(s_\omega\caH_\omega).
\end{align}
From this and (\ref{jj1}), we obtain
\begin{align}
\Theta_{R\to L}(x)=\Theta_{L\to R}(x)={\theta} x{\theta}^*,\quad x\in B(s_\omega\caH_\omega).
\end{align}
Therefore, from (\ref{jj2}),
\begin{align}
J_\omega\lmk x\otimes s_\omega\rmk J_\omega^*
=s_\omega\otimes {\theta} x{\theta}^*,\quad x\in B(s_\omega\caH_\omega).
\end{align}
This proves the second statement.

Furthermore, we have 
\begin{align}
s_\omega\otimes x=J_\omega^2\lmk s_\omega\otimes x\rmk \lmk J_\omega^*\rmk^2
=J_\omega\lmk {\theta} x{\theta}^*\otimes s_\omega \rmk J_\omega^*
=s_\omega\otimes {\theta}^2 x\lmk{\theta}^*\rmk^2,
\end{align}
for any $x\in B(s_\omega\caH_\omega)$.
This means ${\theta}^2=\tilde{\kappa}_{\theta} s_\omega$ with some $\tilde{\kappa}_{\theta}\in\bbT$.
Then by the anti-linearity of ${\theta}$, we have 
\begin{align}
\tilde{\kappa}_{\theta}{\theta}={\theta}^2{\theta}={\theta}{\theta}^2={\theta} \tilde{\kappa}_{\theta}=\bar{\tilde{\kappa}}_{\theta}{\theta}.
\end{align}
This means $\tilde{\kappa}_{\theta}$ is real, namely $\tilde{\kappa}_{\theta}=\pm1$.

This $\tilde{\kappa}_{\theta}$ is independent of the choice of ${\theta}$ satisfying (\ref{jtj}) for 
if ${{\theta_1}}$ is another such anti-unitary, we have
\begin{align}
{{\theta_1}} x{{{\theta_1}}}^*\otimes s_\omega
=
J_\omega\lmk s_\omega\otimes x\rmk J_\omega^*
={\theta} x{\theta}^*\otimes s_\omega,
\end{align}
for all $x\in B(s_\omega\caH_\omega)$.
Hence ${\theta}^*{{\theta_1}}$ is a unitary operator on $s_\omega\caH_\omega$
which commutes with any $x\in B(s_\omega\caH_\omega)$, i.e., ${\theta}=c{{\theta_1}}$ for some
$c\in\bbT$.
We then have 
\begin{align}
\tilde{\kappa}_{\theta} s_\omega={\theta}^2=c{{\theta_1}} c{{\theta_1}}
=c\bar c{{\theta_1}} {{\theta_1}}=\tilde{\kappa}_{{{\theta_1}}}s_\omega,
\end{align}
and get $\tilde{\kappa}_{\theta}=\tilde{\kappa}_{{{\theta_1}}}=:\kappa_\omega$.
This proves the third statement.
\end{proof}
The sign $\kappa_\omega$  coincides with our $\bbZ_2$-index $\sigma_\omega$.
\begin{thm}\label{coffee}
Let $\omega$ be
a reflection invariant pure state on $\caA$ which satisfies the split property with respect to $\caA_R$ and
$\caA_L$. 
Let $(\caH_\omega, \pi_\omega, \Omega_\omega, \Gamma_\omega)$  be a reflection-spilt representation associated to
$\omega$.
Let $s_\omega$ be a projection in $B(\caH_\omega)$ such that 
the support projection of $\Omega_\omega$ in 
$\unit_{\caH_\omega}\otimes B(\caH_\omega)$ is $\unit_{\caH_\omega}\otimes s_\omega$.
Set $\caM:=s_\omega\otimes B(s_\omega\caH_\omega)$, and
$p_\omega=s_\omega\otimes s_\omega$.
Let $J_\omega$ be the modular conjugation on $p_\omega\hat\caH_\omega$
associated to
$(\caM,\Omega_\omega)$.
Let 
$\sigma_\omega$ be the $\bbZ_2$-index associated to $\omega$ in Definition \ref{index} and
$\kappa_\omega$ the $\bbZ_2$-index associated to $\omega$ in Proposition \ref{kdef}.
Then 
we have
$\kappa_\omega=\sigma_\omega$.
\end{thm}
\begin{proof}
We use the notation used in Lemma \ref{thir}.
From $\Gamma_\omega\Omega_\omega=\Omega_\omega$,
we obtain $u=\sigma_\omega c^*u^*c$: first we have
\begin{align}
&\sum_{k\in \Lambda} 
u\rho_\omega^{\frac 12}\zeta_k\otimes \zeta_k
=\sum_{k\in \Lambda} \sqrt{\lambda_k} 
u\zeta_k\otimes \zeta_k
=\Omega_\omega=
\Gamma_\omega\Omega_\omega
=\sum_{k\in \Lambda} \sqrt{\lambda_k} 
\Gamma_\omega\lmk u\zeta_k\otimes \zeta_k\rmk
=\sum_{k\in \Lambda} \sqrt{\lambda_k} 
\sigma_\omega \lmk \zeta_k\otimes u\zeta_k\rmk\nonumber\\
&=\sigma_\omega 
\lmk
u^*\otimes u
\rmk
\sum_{k\in \Lambda} \sqrt{\lambda_k} 
 \lmk u\zeta_k\otimes \zeta_k\rmk
=\sigma_\omega \lmk
u^*\otimes s_\omega
\rmk
\lmk
s_\omega\otimes u
\rmk
\Omega_\omega
=\sigma_\omega \lmk
u^*\otimes s_\omega
\rmk\sum_{l\in\Lambda} u\rho_\omega^{\frac 12} c^* u^*c\zeta_l\otimes \zeta_l\nonumber\\
&=\sigma_\omega \sum_{l\in\Lambda} \rho_\omega^{\frac 12} c^* u^*c\zeta_l\otimes \zeta_l.
\end{align}
Here we used Theorem \ref{indint} and Lemma \ref{thir} {\it 3.}
From this we obtain
\begin{align}
u\rho_\omega^{\frac 12}\zeta_k
=\sigma_\omega \rho_\omega^{\frac 12} c^* u^*c\zeta_k,
\end{align}
for all $k\in\Lambda$.
Hence we have
\begin{align}
u\rho_\omega^{\frac 12}
=\sigma_\omega c^* u^* c c^*u c\rho_\omega^{\frac 12} c^* u^*c.
\end{align}
By the uniqueness of the polar decomposition, we obtain the claim
\begin{align}\label{ucuc}
u=\sigma_\omega c^* u^* c.
\end{align}

Now we are ready to complete the proof of the Theorem.
By Proposition \ref{kdef} and Lemma \ref{thir} {\it 4.}, 
we have $(uc)^2=(uc^*)^2=\kappa_\omega s_\omega$.
From this and (\ref{ucuc}),
 we have
\begin{align}
\kappa_\omega
=\braket{(uc)^2\zeta_k}{\zeta_k}
=\braket{ucuc\zeta_k}{\zeta_k}
=\braket{(uc)^*\zeta_k}{uc\zeta_k}
=\braket{c^*u^*c\zeta_k}{u\zeta_k}
=\braket{\sigma_\omega u\zeta_k}{u\zeta_k}
=\sigma_\omega,
\end{align}
for any $k\in\Lambda$.
This completes the proof.
\end{proof}
\begin{rem}\label{su}
From the proof, we see that one way to derive the index $\sigma_\omega$ for concrete state
$\omega$ is to consider the Schmidt decomposition and calculate $u$.
Using (\ref{ucuc}), we can obtain $\sigma_\omega$.
\end{rem}
\section{$\bbZ_2$-index in Matrix Product States}\label{poco}
In this section, we prove that the $\bbZ_2$-index $\sigma_\omega$
for a matrix product state $\omega$ is the same as the $\bbZ_2$-index found in \cite{po}.
First let us recall known facts on matrix product states.
Let $k\in\bbN$ be a number
and $\vv=(v_{\mu})_{\mu=1,\ldots, d}\in \Mat_k^{\times d}$ a $d$-tuple of
$k\times k$ matrices.
For each $l\in\nan$, we set
\begin{equation}
{\mathcal K}_l(\vv) :=\spn\left\{v_{\mu_0}v_{\mu_{1}}\ldots v_{\mu_{l-1}}\mid
(\mu_0,\mu_1,\ldots,\mu_{l-1})\subset\{1,\ldots, d\}^{\times l}\right\}.
\end{equation}
We say $\vv$ is primitive if ${\mathcal K}_l(\vv)=\Mat_k$ for $l$ large enough.
We denote by $\Primz_u(d,k)$ the set of 
all primitive $d$-tuples $\vv$ of $k\times k$ matrices
which are normalized, i.e.,
\[
\sum_{\mu=1,\ldots, d} v_\mu v_\mu^*=1.
\]

For $\vv\in\Primz_u(d,k)$, there exists a unique $T_\vv$-invariant state
$\hat \rho_\vv$, and it is faithful. (See \cite{Wolf:2012aa} for example.)We denote the density matrix corresponding to
$\hat \rho_\vv$ by $\rho_\vv$.
Each $\vv\in \Primz_u(d,k)$ generates a translationally invariant state $\omega_\vv$ by
\begin{align}
\omega_{\vv}\lmk
\bigotimes_{i=0}^{l-1}
\ketbra{\psi_{\mu_i}}{\psi_{\nu_i}}
\rmk=
\hat\rho_\vv\lmk v_{\mu_0}\cdots v_{\mu_{l-1}} v_{\nu_{l-1}}^*\cdots v_{\nu_0}^*\rmk,\quad
\mu_i,\nu_i=1,\ldots,d,\quad i=0,\ldots,l-1,\quad l\in\nan. 
\end{align}
A translationally invariant state which has this representation is called a matrix product state.
This representation is unique up to
unitary and phase \cite{fnwpure}:
If both of $\vv^{(1)}\in \Primz_u(d,k_1)$ and $\vv^{(2)}\in \Primz_u(d,k_2)$
generate the same matrix product state, then $k_1=k_2$
and there exist a unitary $U:\cc^{k_1}\to\cc^{k_2}$ and $e^{i\theta}\in\bbT$
such that
\begin{align}\label{unique}
U v_{\mu}^{(1)}=e^{i\theta} v_{\mu}^{(2)}U,\quad
\mu=1,\ldots,d.
\end{align}
%(See \cite{fnwpure}, \cite{Ogata3}.)

Let $\omega$ be a reflection invariant matrix product state generated by $\vv\in \Primz_u(d,k)$.
It is a unique ground state of some translation invariant finite range interaction.
 i.e., there is an interaction
$\Phi_\vv$ given by some fixed local positive element $h_\vv\in\caA_{[0,m-1]}$ with some
$m\in\nan$ as
\begin{align}\label{hamdef}
\Phi_{\vv}(X):=\left\{
\begin{gathered}
\beta_x\lmk h_\vv\rmk,\quad \text{if}\quad  X=[x,x+m-1]\cap\bbZ \quad \text{for some}\quad  x\in\bbZ\\
0,\quad\text{otherwise}
\end{gathered}\right.
\end{align}
for each $X\in {\mathfrak S}_{\bbZ}$ and $\omega$ is a unique $\tau^{\Phi_\vv}$-ground state. (See \cite{Fannes:1992vq} and \cite{Ogata3} Corollary 5.6 \cite{Ogata1} Theorem 1.18, Lemma 3.25.)
For this interaction $h_\vv$, $1-h_{\vv}$ is equal to 
the support projection of $\omega\vert_{\caA_{[0,m-1]}}$. (
Note that $\left. \omega\right\vert_{\caA_L}=\Xi_L(\hat \rho_\vv)$ with the $\Xi_L$ in Lemma 3.14 of \cite{Ogata1} and the $T_\vv$-invariant state $\hat \rho_\vv$.
From  the proof of Lemma 3.19 of \cite{Ogata1} equation (48), we see that
$1-h_{\vv}$ is equal to 
the support projection of $\omega\vert_{\caA_{[0,m-1]}}$. Note that primitive $\vv$ belongs to $\ClassA$, Remark 1.16 of \cite{Ogata1}).
Therefore, from the reflection invariance of $\omega$, $\Phi_\vv$ is reflection invariant.
%For each $\vv\in \Primz_u(d,k)$, we fix such $h_\vv$.
The Hamiltonian given by this interaction is frustration-free,
i.e.,
for each finite interval $I$ with $|I|\ge m$,
the local Hamiltonian $\lmk H_{\Phi_\vv}\rmk_{I}$ has a nontrivial kernel, which is the ground state space
of $\lmk H_{\Phi_\vv}\rmk_{I}$.
We denote by $G_{I,\vv}$, the orthogonal projection onto this kernel.
By Lemma 3.19 of \cite{Ogata1}, and its proof (equation (48)),
the support of the restriction $\left.\omega\right\vert_{\caA_I}$
is equal to $G_{I,\vv}$ and  there exists some constant $d_\vv>0$ such that
\begin{align}\label{fb}
\psi\le d_\vv \cdot \omega,
\end{align}
for any frustration free state $\psi$ on $\caA_R$, i.e., a state $\psi$ satisfying $\psi(\beta_x(h_\vv))=0$ for any
$0\le x\in\bbZ$.
(See proof of Lemma 2.3 of \cite{Ogata2}.)

%For each $l\in\nan$,
%define
%by
%\begin{align}\label{eq:gbdef}
%\Gamma^{(R)}_{l,\vv}\lmk X\rmk
%=\sum_{\mu^{(l)}\in \{1,\cdots,n\}^{\times l}}\lmk\Tr X \lmk\widehat{v_{\mu^{(l)}}} \rmk^*\rmk\widehat{\psi_{\mu^{(l)}}},\quad
%X\in\mk,
%\end{align}
%and set
%$
%\cgv{l}:=\Ran\Gamma^{(R)}_{l,\vv}\subset \bigotimes_{i=0}^{l-1}\cc^n.
%$
%Furthermore, we denote by
%$G_{l,\vv}$
%the orthogonal projection onto $\cgv{l}$ in $\bigotimes_{i=0}^{l-1}\cc^n$.
%We set $h_{m,\vv}:=1-G_{m,\vv}$ and $\Phi_{m,\bb}:=\Phi_{h_{m,\bb}}$.
%For the simplicity of the terminology, 
%we use the symbol 
%$\bbm_{\vv}$
%to denote $\bbm_{\{\cgv{N}\}_N}$.
%
%
%
%
Now let us come back to our problem.
With the analogous argument as in \cite{po}, we obtain the following.
See \cite{TasakiBook} for a nice description.
\begin{lem}\label{lem5}
Let $\omega$ be a reflection invariant matrix product state generated by
$\vv\in  \Primz_u(d,k)$. Let $\hat \rho_\vv$ be the $T_\vv$-invariant state
given by a density matrix $\rho_\vv$.
Then there exist $e^{it}\in\bbT$ and 
$\theta:\bbC^k\to\bbC^k$ an anti-unitary
such that
\begin{align}\label{vt}
v_\mu=e^{it}\rho_\vv^{-\frac 12} \theta v_\mu^* \theta^* \rho_\vv^{\frac 12},\quad
\mu=1,\ldots, d.
\end{align}
For any $e^{it}\in\bbT$ and 
$\theta:\bbC^k\to\bbC^k$ an anti-unitary
satisfying (\ref{vt}), we have
\begin{align}
\theta^2=\zeta_\omega\unit,
\end{align}
with some $\zeta_\omega\in\{\pm 1\}$.
The value $\zeta_\omega$ does not depend on the choice of $\vv$,
 $e^{it}\in\bbT$ and 
$\theta:\bbC^k\to\bbC^k$.
\end{lem}
\begin{defn}
By this Lemma, we define a $\bbZ_2$-index $\zeta_\omega$.
\end{defn}
\begin{proof}
Let $\rho_\vv=\sum_{j=1}^k{\lambda_j}\ketbra{\xi_j}{\xi_j}$
be the spectral decomposition of $\rho_\vv$, where
$\lambda_j>0$ and $\{\xi_j\}_{j=1}^d$ is an orthonormal basis of $\bbC^k$.
Let $c:\bbC^k\to\bbC^k$ be the complex conjugation such that $c\xi_j=\xi_j$
for all $j=1,\ldots,k$. Note that $c^*=c$.

Set 
\begin{align}
\bar v_\mu:=c v_\mu c,\quad
\tilde v_\mu
:=\rho_\vv^{-\frac 12} \lmk \bar v_\mu\rmk^* \rho_\vv^{\frac 12},\quad
\mu=1,\ldots, d.
\end{align}
We claim $\tilde \vv \in \Primz_u(d,k)$  and it generates $\omega$.
Since ${\mathcal K}_l(\vv)=\Mat_k$ for $l$ large enough,
${\mathcal K}_l(\tilde \vv)=\Mat_k$ for $l$ large enough.
Hence $\tilde\vv$ is primitive.
Using $c\rho_\vv c=\rho_\vv$ and $\sum_{\mu} v_\mu^* \rho_\vv  v_\mu =\rho_\vv$, we have
\begin{align}
\sum_{\mu}\tilde v_\mu \tilde v_\mu^*
=\sum_{\mu}\rho_\vv^{-\frac 12} c v_\mu^*c \rho_\vv c v_\mu c \rho_\vv^{-\frac 12}
=\sum_{\mu}\rho_\vv^{-\frac 12} c v_\mu^* \rho_\vv  v_\mu c \rho_\vv^{-\frac 12}
=\rho_\vv^{-\frac 12} c\rho_\vv  c \rho_\vv^{-\frac 12}
=\unit_k
\end{align}
Hence we have $\tilde \vv \in \Primz_u(d,k)$.
The state $\hat\rho_\vv$ is $T_{\tilde \vv}$-invariant because
\begin{align}
\sum_{\mu} \tilde v_\mu^* \rho_\vv \tilde v_\mu
=\sum_\mu \rho_\vv^{\frac 12 }cv_{\mu} c \rho_\vv^{-\frac 12} \rho_\vv\rho_\vv^{-\frac 12}
c v_\mu^* c \rho_\vv^{\frac 12}=\rho_\vv,
\end{align}
from $\sum_\mu v_\mu v_\mu^*=1$.
Now we show that $\tilde \vv$ generates $\omega$.
For any $l\in\nan$ and $\mu_i,\nu_i=1,\ldots d$, $i=0,\ldots, l-1$,
from the reflection invariance and translation invariance, we have
\begin{align}\label{og}
&\omega\lmk
\bigotimes_{i=0}^{l-1} \ketbra{\psi_{\mu_i}}{\psi_{\nu_i}}
\rmk
=\omega\circ\gamma\lmk
\bigotimes_{i=0}^{l-1} \ketbra{\psi_{\mu_i}}{\psi_{\nu_i}}
\rmk
=\omega\lmk
\bigotimes_{i=-l}^{-1} \ketbra{\psi_{\mu_{-i-1}}}{\psi_{\nu_{-i-1}}}
\rmk
=\omega\lmk
\bigotimes_{i=0}^{l-1} \ketbra{\psi_{\mu_{l-i-1}}}{\psi_{\nu_{l-i-1}}}
\rmk\nonumber\\
&=\hat\rho_\vv\lmk
v_{\mu_{l-1}}v_{\mu_{l-2}}\cdots v_{\mu_1}v_{\mu_0}v_{\nu_0}^* v_{\nu_1}^*
\cdots v_{\nu_{l-1}}^*
\rmk
=\sum_{j=1}^k \lambda_j 
\braket{\xi_j}{ v_{\mu_{l-1}}v_{\mu_{l-2}}\cdots v_{\mu_1}v_{\mu_0}v_{\nu_0}^* v_{\nu_1}^*
\cdots v_{\nu_{l-1}}^* \xi_j}
\end{align}
Note that 
\begin{align}
&\braket{\xi_j}{ v_{\mu_{l-1}}v_{\mu_{l-2}}\cdots v_{\mu_1}v_{\mu_0}v_{\nu_0}^* v_{\nu_1}^*
\cdots v_{\nu_{l-1}}^* \xi_j}
=\braket{c\xi_j}{ v_{\mu_{l-1}}v_{\mu_{l-2}}\cdots v_{\mu_1}v_{\mu_0}v_{\nu_0}^* v_{\nu_1}^*
\cdots v_{\nu_{l-1}}^* c\xi_j}\nonumber\\
&
%\braket{ cv_{\mu_{l-1}}v_{\mu_{l-2}}\cdots v_{\mu_1}v_{\mu_0}v_{\nu_0}^* v_{\nu_1}^*
%\cdots v_{\nu_{l-1}}^* c\xi_j}{\xi_j}
=\braket{\xi_j}{ cv_{\nu_{l-1}}v_{\nu_{l-2}}\cdots v_{\nu_1}v_{\nu_0}v_{\mu_0}^* v_{\mu_1}^*
\cdots v_{\mu_{l-1}}^* c\xi_j}=\braket{\xi_j}{ \bar v_{\nu_{l-1}}\bar v_{\nu_{l-2}}\cdots \bar v_{\nu_1}\bar v_{\nu_0}\bar
v_{\mu_0}^* \bar v_{\mu_1}^*
\cdots \bar v_{\mu_{l-1}}^* \xi_j}.
\end{align}
Substituting this to (\ref{og}), we have
\begin{align}
&(\ref{og})
=\hat\rho_\vv\lmk
\bar v_{\nu_{l-1}}\bar v_{\nu_{l-2}}\cdots \bar v_{\nu_1}\bar v_{\nu_0}\bar
v_{\mu_0}^* \bar v_{\mu_1}^*
\cdots \bar v_{\mu_{l-1}}^*
\rmk
=\Tr\lmk
\rho_\vv^{\frac 12}
\lmk
\bar v_{\nu_{l-1}}\bar v_{\nu_{l-2}}\cdots \bar v_{\nu_1}\bar v_{\nu_0}
\rho_\vv^{-\frac 12}\rho_\vv\rho_\vv^{-\frac 12}
\bar
v_{\mu_0}^* \bar v_{\mu_1}^*
\cdots \bar v_{\mu_{l-1}}^*
\rmk\rho_\vv^{\frac 12}
\rmk\nonumber\\
&=
\hat\rho_\vv
\lmk
\tilde
v_{\mu_0} \tilde v_{\mu_1}
\cdots \tilde v_{\mu_{l-1}}
\tilde v_{\nu_{l-1}}^*\tilde v_{\nu_{l-2}}^*\cdots \tilde v_{\nu_1}^*\tilde v_{\nu_0}^*
\rmk.
\end{align}
Hence  $\tilde \vv$ generates $\omega$, proving the claim.

Now, as both of $\vv$ and $\tilde\vv$ generates same state $\omega$, by the uniqueness (\ref{unique}),
there exist a unitary $U$ on $\bbC^k$ and $e^{it}\in\bbT$
such that
\begin{align}\label{uvvu}
U v_\mu=e^{it}\tilde v_\mu U,\quad \mu=1,\ldots, d.
\end{align}
From the fact that $\hat\rho_\vv$ is $T_{\tilde \vv}$-invariant 
and (\ref{uvvu}), we see that the state $\hat\rho_\vv\circ \Ad U$ is $T_\vv$-invariant.
By the uniqueness of $T_\vv$-invariant state (by $\vv\in \Primz_u(d,k)$),
we get 
\begin{align}\label{ur}
U\rho_\vv U^*=\rho_\vv.
\end{align}
Set $\theta:=U^* c:\bbC^k\to\bbC^k$ an anti-unitary operator.
From (\ref{uvvu}) and the definition of $\tilde\vv$, and (\ref{ur}),
we obtain (\ref{vt}):
\begin{align}
v_\mu=e^{it} U^*\tilde v_\mu U
=e^{it}U^*\rho_\vv^{-\frac 12} \lmk \bar v_\mu\rmk^* \rho_\vv^{\frac 12}U
=e^{it}\rho_\vv^{-\frac 12} ( U^* c) v_\mu^* (U^*c)^*\rho_\vv^{\frac 12}
=e^{it}\rho_\vv^{-\frac 12} \theta v_\mu^*\theta^*\rho_\vv^{\frac 12}.
\end{align}

Now for any
$e^{it_0}\in\bbT$ and 
$\theta_0:\bbC^k\to\bbC^k$ an anti-unitary
satisfying 
\begin{align}
v_\mu=e^{it_0}\rho_\vv^{-\frac 12} {\theta_0} v_\mu^* {\theta_0}^* \rho_\vv^{\frac 12},\quad
\mu=1,\ldots, d,
\end{align}
we show
${\theta_0}^2=\unit$ or ${\theta_0}^2=-\unit$. 
Taking adjoint of (\ref{vt}), we have
\begin{align}
v_\mu^*
=e^{-it_0}\rho_\vv^{\frac 12} {\theta_0} v_\mu{\theta_0}^*\rho_\vv^{-\frac 12}.
\end{align}
Substituting this to (\ref{vt}), we obtain
\begin{align}
v_\mu=e^{it_0}\rho_\vv^{-\frac 12} {\theta_0} 
e^{-it_0}\rho_\vv^{\frac 12} {\theta_0} v_\mu{\theta_0}^*\rho_\vv^{-\frac 12}
{\theta_0}^*\rho_\vv^{\frac 12}
=e^{2it_0}\rho_\vv^{-\frac 12} {\theta_0} 
\rho_\vv^{\frac 12} {\theta_0} v_\mu{\theta_0}^*\rho_\vv^{-\frac 12}
{\theta_0}^*\rho_\vv^{\frac 12}
\end{align}
Since $\vv$ is primitive, this means 
$e^{2it_0}=1$ and 
$\rho_\vv^{-\frac 12} {\theta_0} 
\rho_\vv^{\frac 12}{\theta_0} =b\unit$ for some $b\in\bbC$.
Decomposing $b=e^{is}|b|$ with $e^{is}\in\bbT$, we have
\begin{align}
{\theta_0}^2 \lmk {\theta_0}^* \rho_\vv^{\frac 12} {\theta_0}\rmk
=e^{is}|b|\rho_\vv^{\frac 12}.
\end{align}
By the uniqueness of the polar decomposition and the faithfulness of $\rho_\vv$,
we get ${\theta_0}^2=e^{is}\unit$.
But then 
\begin{align}
e^{is}{\theta_0}={\theta_0}^2{\theta_0}={\theta_0}^3={\theta_0}{\theta_0}^2
={\theta_0} e^{is}=e^{-is}{\theta_0}.
\end{align}
Therefore, $e^{is}$ is real and we get that ${\theta_0}^2=\unit$ or ${\theta_0}^2=-\unit$. 

To prove the independence of this sign of
$\vv$,
 $e^{it}$ and 
$\theta$,
let 
$\oo\in  \Primz_u(d,k')$ be a generator of $\omega$, with a $T_\oo$-invariant state
$\hat \rho_\oo$ given by a density matrix $\rho_\omega$. Let
$e^{iu}\in\bbT$ and 
$\xi:\bbC^{k'}\to\bbC^{k'}$ an anti-unitary
such that
\begin{align}\label{vto}
\omega_\mu=e^{iu}\rho_\oo^{-\frac 12} \xi \omega_\mu^* \xi^* \rho_\oo^{\frac 12}.
\end{align}
Since both of $\vv$ and $\oo$ generates same $\omega$, from the uniqueness (\ref{unique}),
$k=k'$ and
there exist a unitary $V$ on $\bbC^k$ and $e^{i\lambda}\in\bbT$ such that
$v_\mu=e^{i\lambda }V^* \omega_\mu V$.
From
 \begin{align}
\hat\rho_\vv= \hat \rho_\vv\circ T_\vv
=\hat \rho_\vv\circ \Ad (V^*)\circ T_\oo\circ \Ad(V),
 \end{align}
 $\hat \rho_\vv\circ \Ad (V^*)$ is a $T_\oo$-invariant state.
 By the uniqueness of a $T_\oo$-invariant state,
 we get
 $\rho_\oo=V \rho_\vv V^*$.
 Now we have
 \begin{align}
 e^{i\lambda }V^* \omega_\mu V
 =v_\mu=e^{it}\rho_\vv^{-\frac 12} \theta v_\mu^* \theta^* \rho_\vv^{\frac 12}
 =e^{it}\rho_\vv^{-\frac 12} \theta 
  \lmk e^{i\lambda }V^* \omega_\mu V\rmk^*
  \theta^* \rho_\vv^{\frac 12}
 =e^{it+i\lambda} \rho_\vv^{-\frac 12} \theta 
  V^* \omega_\mu^* V
  \theta^* \rho_\vv^{\frac 12}.
 \end{align} 
 From this, (\ref{vto}), and $\rho_\oo=V \rho_\vv V^*$,
 we have
\begin{align}
 e^{it} \rho_\oo^{-\frac 12} V\theta 
  V^* \omega_\mu^* V
  \theta^* V^*\rho_\oo^{\frac 12}
=
 e^{it} V\rho_\vv^{-\frac 12} \theta 
  V^* \omega_\mu^* V
  \theta^* \rho_\vv^{\frac 12}V^*
  =\omega_\mu
  =e^{iu}\rho_\oo^{-\frac 12} \xi \omega_\mu^* \xi^* \rho_\oo^{\frac 12}.
\end{align}
Therefore, we get
 \begin{align}
 e^{it} V\theta 
  V^* \omega_\mu^* V
  \theta^* V^*
  =e^{iu} \xi \omega_\mu^* \xi^* .
\end{align}
Since $\oo$ is primitive, this means 
$\xi^* V\theta 
  V^*=e^{i\eta}\unit$ for some $e^{i\eta}\in \bbT$.
  Then we have
  \begin{align}
  V \theta^2V^*=V\theta V^* V\theta V^*
  =e^{-i\eta}\xi e^{-i\eta}\xi
  =\xi^2 .
  \end{align}
  This proves the claim.
\end{proof}

Since a matrix product state $\omega$ generated by a normalized primitive $d$-tuple
is a unique gapped ground state by \cite{Fannes:1992vq},\cite{Ogata3}, 
 it is pure and satisfies the split property.
Therefore, if furthermore $\omega$ is reflection invariant, 
we can associate $\omega$, our $\bbZ_2$-index $\sigma_\omega$ in Definition \ref{index}.
We then have the following theorem.
\begin{thm}
For a reflection invariant matrix product state $\omega$ generated by a normalized primitive $d$-tuple of matrices, we have
\[
\sigma_\omega=\zeta_\omega.
\]
\end{thm}
\begin{proof}Let $\omega$ be a reflection
 invariant matrix product state generated by $\vv\in \Primz_u(d,k)$.
Let $(\caH_\omega, \pi_\omega, \Omega_\omega, \Gamma_\omega)$  be a reflection-spilt representation associated to
$\omega$.
Then
\[
\lmk \hat \caH_\omega:=\caH_\omega\otimes \caH_\omega,
\hat\pi_\omega:=\lmk \pi_\omega\circ\gamma_{L\to R}\rmk\otimes \pi_\omega,
\Omega_\omega
\rmk
\]
is a GNS triple of $\omega$.
Let $s_\omega$ be a projection in $B(\caH_\omega)$ such that 
the support projection of $\Omega_\omega$ in 
$\unit_{\caH_\omega}\otimes B(\caH_\omega)$ is $\unit_{\caH_\omega}\otimes s_\omega$.
Set $\caM:=s_\omega\otimes B(s_\omega\caH_\omega)$, and
$p_\omega=s_\omega\otimes s_\omega$.
We define a density matrix $\rho_\omega$ on $\caH_\omega$
by
\begin{align}\label{rhod}
\Tr_{\caH_\omega}\lmk\rho_\omega x\rmk=
\braket{\Omega_\omega}{\lmk
\unit_{\caH_\omega}\otimes x\rmk
\Omega_\omega},\quad
x\in B(\caH_\omega).
\end{align}

Since $\omega$ is translation invariant, there is a unitary $V$ on $\hat \caH_\omega$
such that 
\[
V\hat \pi_\omega(A)V^*
=\hat\pi_\omega\circ \beta_1\lmk A\rmk,\quad
A\in \caA.
\]
From this,
we obtain a homomorphism from
$\hat\pi_\omega(\caA_R)''=\unit_{\caH_\omega}\otimes B(\caH_\omega)$
onto $\lmk \hat\pi_\omega\circ\beta_1(\caA_R)\rmk''\subset \unit_{\caH_\omega}\otimes B(\caH_\omega)$
\begin{align}
\unit_{\hat H_\omega}\otimes B(\caH_\omega)\ni \unit_{\caH_\omega} \otimes x
\mapsto 
V \lmk \unit_{\caH_\omega} \otimes x\rmk V^*
\in \unit_{H_\omega}\otimes B(\caH_\omega).
\end{align}
Therefore, there exists an endomorphism $\Theta$ on $B(\caH_\omega)$
such that
\begin{align}\label{theta}
V \lmk \unit_{\caH_\omega} \otimes x\rmk V^*
=\unit_{\caH_\omega} \otimes\Theta(x),\quad x\in B(\caH_\omega),
\end{align}
and
$\Theta(B(\caH_\omega))' =\pi_{\omega}(\Mat_d\otimes\unit_{\caA_{[1,\infty)}})$.
(Recall Lemma 2.6.8 of \cite{BR1}.)
Note that $\Theta\circ\pi_\omega(A)=\pi_\omega\circ\beta_1(A)$
for any $A\in\caA_R$.
We recall the following fact:
\begin{lem}[\cite{arv}]\label{lem:arv}
Let $\caH$ be a separable infinite dimensional Hilbert space, and $n\in \nan$.
Let $\Phi : B(\caH)\to B(\caH)$ be a
unital endomorphism of $B(\caH)$
such that$
\lmk \Phi\lmk B(\caH)\rmk\rmk'
$
is isomorphic to $\Mat_n$.
Let $\{E_{ij}\}_{i,j=1,\ldots, n}$
be a system of matrix units of  $
\lmk \Phi\lmk B(\caH)\rmk\rmk'
$.
Then there exist $S_i\in B(\caH)$, $i=1,\ldots,n$ such that
\begin{align}\label{eq:cuntz}
S_i^*S_j=\delta_{ij},\quad E_{ij}=S_iS_j^*,\quad 
\sum_{j=1}^n S_j x S_j^*=\Phi(x),\quad x\in B(\caH).
\end{align}
\end{lem}
Applying this to our $\Theta$ in (\ref{theta}),
we obtain
operators $\cs_\mu\in B(\caH_\omega)$ with $\mu=1,\ldots d$
satisfying the following:
\begin{align}
&\cs_{\mu}^*\cs_{\nu}=\delta_{\mu\nu}\unit,\label{s1}\\
&\sum_{\mu=1,\ldots,d} \cs_{\mu} \pi_{\omega}(A) \cs_{\mu}^*=\pi_\omega\circ\beta_{1}(A),\quad A\in \caA_R \label{st}.\\
&\pi_\omega\lmk e_{\mu\nu}\otimes\unit_{[1,\infty)}\rmk
=\cs_\mu \cs_\nu^*\quad
\text{for all} \quad \mu,\nu=1,\ldots,d.\label{sr}
\end{align}
(See \cite{arv,bjp,BJ}, Proof of Proposition 3.5 of \cite{Matsui1} and Lemma 3.5 of \cite{Matsui3}.)
Here $ e_{\mu\nu}\otimes\unit_{[1,\infty)}$ indicates an element $e_{\mu\nu}$ in $\caA_{\{0\}}=\Mat_{d}$
embedded into $\caA_R$.
From (\ref{s1}), (\ref{st}) and (\ref{sr}), we have
\begin{align}\label{ssss}
\pi_\omega\lmk\bigotimes_{k=0}^{l-1}e_{\mu_k,\nu_k}\rmk
=\cs_{\mu_0}\cdots \cs_{\mu_{l-1}}\cs_{\nu_{l-1}}^*\cdots \cs_{\nu_0}^*,
\end{align}for all $l\in\nan$, $\mu_k,\nu_k=1,\ldots,d$.

Now we restrict these $\cs_\mu$ to {\it a frustration-free subspace}
$\caK$ of $\caH_\omega$.
Recall that $\omega$ is the frustration free ground state of the translation invariant finite range interaction $\Phi_\vv$ (\ref{hamdef}). Namely, there is a self-adjoint element $h_\vv\in \caA_{[0,m-1]}$, 
such that
$\omega(\beta_x(h_\vv))=0$ for all $x\in\bbZ$.
We consider the following frustration-free
subspace of $\caH_\omega$:
\begin{align*}
\caK:=\cap_{\bbZ\ni x\ge 0}\ker \pi_\omega\lmk
\beta_x\lmk h_\vv\rmk
\rmk.
\end{align*}
Note that the support of $\rho_\omega$ defined in  (\ref{rhod}),  is in $\caK$, because $\omega$ is frustration-free.
Let $P_\caK$ be the orthogonal projection onto $\caK$.
As in \cite{Matsui3} (Lemma 3.2 and the argument in the proof of Lemma 3.6 ), $\caK$ is a finite dimensional space,
and 
$\cs_\mu^*$ preserves $\caK$:
\begin{align}\label{psb}
\cs_\mu^* P_\caK=P_{\caK}\cs_\mu^* P_\caK ,\quad \mu=1,\ldots,d.
\end{align}
We denote $(\cs_\mu^* P_\caK)^*\in B(\caK)$ by $B_\mu$, $\mu=1,\ldots,d$.
Note that $\rho_\omega$ is of finite rank because $\caK$ is finite dimensional.

We claim that $\bbB=(B_\mu)_{\mu=1,\ldots,d}\in \Primz_u(d,\dim\caK)$.To prove this, it suffices to show that
${\rho_\omega}$ is faithful on $\caK$ and for the completely positive unital map $T_\bbB$ defined
by $T_\bbB(x)=\sum_{\mu=1,\ldots,d} B_\mu x B_\mu^*$, $x\in B(\caK)$,
we have $T_\bbB^N(x)\to \Tr_{\caH_\omega}\lmk \rho_\omega x\rmk\unit $, as $N\to\infty$,
for each $x\in B(\caK)$. (See Lemma C.5 of \cite{Ogata1}.)
First we show that ${\rho_\omega}$ is faithful on $\caK$.
If ${\rho_\omega}$ is not faithful on $\caK$, 
then there exists a unit vector $\xi\in\caK$ which is orthogonal to
the support of ${\rho_\omega}$.
By the definition of $\caK$, this $\xi$ defines a frustration free state
$\psi=\braket{\xi}{\pi_\omega\lmk\cdot\rmk\xi}$ on $\caA_R$.
Let $p$ be the orthogonal projection onto the one-dimensional space $\bbC\xi$.
As $\pi_\omega(\caA_R)''=B(\caH_\omega)$, by Kaplansky's density Theorem, (Theorem 2.4.16 of \cite{BR1}) there exists
a net $\{x_\alpha\}_{\alpha}$ of positive elements in the unit ball of $\caA_{R}$
such that $\pi_\omega\lmk x_\alpha\rmk\to p$ in the $\sigma w$-topology.
For this net, we have $\lim_{\alpha}\omega(x_\alpha)=0$
and $\lim_\alpha\psi(x_\alpha)=1$.
This contradicts to (\ref{fb}).
Hence ${\rho_\omega}$ is faithful on $\caK$.
Next we show $T_\bbB^N(x)\to \Tr_{\caH_\omega}\lmk \rho_\omega x\rmk\unit $, as $N\to\infty$ for all $x\in B(\caK)$.
By $\pi_\omega(\caA_R)''=B(\caH_\omega)$ and the finite dimensionality of $\caK$,
we have $B(\caK)=P_\caK \pi_\omega\lmk \caA_{R}\cap\caA_{\rm loc}\rmk P_\caK$.
Therefore, for each $x\in B(\caK)$, there is an element $A\in \caA_{R}\cap\caA_{\rm loc}$
such that $x=P_{\caK}\pi_\omega\lmk A \rmk P_\caK$.
Since $\omega$ is a factor state and translation invariant, we have
$\sigma w-\lim_{N\to\infty}\pi_\omega\circ \beta_N(A)=\omega(A)\unit$.
Therefore, for any $\eta\in\caK$, we have
\begin{align}
\braket{\eta}{T_{\bbB}^N\lmk x\rmk\eta}
=
\braket{\eta}{T_{\bbB}^N\lmk P_{\caK}\pi_\omega \lmk A\rmk P_{\caK}\rmk\eta}
=\braket{\eta}{\pi_\omega\circ\beta_N\lmk A\rmk\eta}
\to \omega(A)\lV \eta\rV^2
=\Tr_{\caH_\omega}\lmk \rho_\omega x\rmk\lV \eta\rV^2,\quad N\to\infty.
\end{align}
Hence $\bbB\in \Primz_u(d,\dim\caK)$.

The above proof for the primitivity also tells us that $\rho_\omega$ is the $T_\bbB$-invariant state.
From (\ref{ssss}) and the definition of $\bbB$ and (\ref{psb}), we see that 
$\bbB$ is a $d$-tuple generating $\omega$.
Furthermore, as $\rho_\omega$ is faithful on $\caK$, we have $s_\omega=P_\caK$.

Let $J_\omega$ (resp. $\Delta_\omega$) be the modular conjugation
(resp. modular operator) on $p_\omega\hat\caH_\omega$ associated to
$(\caM,\Omega_\omega)$.
By Proposition \ref{kdef}
there exists an anti-unitary $\theta: s_\omega\caH_\omega\to s_\omega\caH_\omega$
such that 
\begin{align}\label{jtheta}
J_\omega\lmk s_\omega\otimes x\rmk J_\omega^*
=\theta x\theta^*\otimes s_\omega,\quad
J_\omega\lmk x\otimes s_\omega\rmk J_\omega^*
=s_\omega\otimes \theta x\theta^*,
\end{align}
for all $x\in B(s_\omega\caH_\omega)$.
By Proposition \ref{kdef} and Theorem \ref{coffee}, we have
$\theta^2=\kappa_\omega s_\omega=\sigma_\omega s_\omega$.
Recall also from Lemma \ref{thir}
\begin{align}\label{dac}
\Delta_\omega^{\frac 12}\lmk s_\omega\otimes x\rmk\Omega_\omega
=\lmk s_\omega\otimes \rho_\omega^{\frac 12}x\rho_\omega^{-\frac 12}\rmk\Omega_\omega,\quad
x\in B(s_\omega\caH_\omega).
\end{align}

Now we prove that for the $\theta$ in (\ref{jtheta}), there exist an $a\in\bbT$
such that
 \begin{align}\label{bget}
a\rho_\omega^{-\frac 12} \theta  B_\nu^*\theta^* \rho_\omega^{\frac 12}
= B_\nu,\quad \nu=1,\ldots, d.
\end{align}
From this, we obtain the claim of the Theorem:
\begin{align}
\sigma_\omega s_\omega=\theta^2=\zeta_\omega s_\omega.
\end{align}

To show (\ref{bget}), we first show
\begin{align}\label{ss4}
\ketbra
{\lmk \cs_\nu^* \otimes \unit\rmk
\Omega_\omega}{\lmk \cs_\mu^*\otimes \unit\rmk\Omega_\omega}
=\ketbra{\lmk  \unit\otimes\cs_\nu^* \rmk
\Omega_\omega}
{\lmk
\unit\otimes \cs_\mu^*
\rmk
\Omega_\omega},\quad \mu,\nu=1,\ldots, d.
\end{align}

 For any $l\in\nan$, $\mu,\nu, \mu_{-l},\ldots, \mu_{-1}, \nu_{-l},\ldots, \nu_{-1}, \lambda_0,\ldots,
\lambda_{l-1}, \eta_0,\ldots, \eta_{l-1}=1,\ldots,d$, we have
\begin{align}\label{nagaii}
&\braket{\lmk \cs_\mu^*\otimes \unit\rmk\Omega_\omega}{
\lmk
\pi_\omega\circ{\gamma_{L\to R}}\lmk\bigotimes_{j=-l}^{-1}e_{\mu_j,\nu_j}\rmk
\otimes \pi_\omega\lmk\bigotimes_{j=0}^{l-1}e_{\lambda_j,\eta_j}
\rmk
\rmk
\lmk \cs_\nu^* \otimes \unit\rmk\Omega_\omega}\nonumber\\
&=\braket{\Omega_\omega}{
\lmk
\cs_\mu \cs_{\mu_{-1}}\cdots \cs_{\mu_{-l}}\cs_{\nu_{-l}}^*\cdots \cs_{\nu_{-1}}^*\cs_\nu^*\otimes \pi_\omega\lmk\bigotimes_{j=0}^{l-1}e_{\lambda_j,\eta_j}\rmk
\rmk
\Omega_\omega}\nonumber\\
&=\braket{\Omega_\omega}{
\lmk
\pi_\omega\circ{\gamma_{L\to R}}\lmk\bigotimes_{j=-l-1}^{-2}e_{\mu_{j+1},\nu_{j+1}}\otimes e_{\mu\nu}\rmk
\otimes \pi_\omega\lmk\bigotimes_{j=0}^{l-1}e_{\lambda_j,\eta_j}
\rmk
\rmk
\Omega_\omega}\nonumber\\
&=\omega
\lmk
\bigotimes_{j=-l-1}^{-2}e_{\mu_{j+1},\nu_{j+1}}\otimes e_{\mu\nu}
\otimes \bigotimes_{j=0}^{l-1}e_{\lambda_j,\eta_j}
\rmk
\end{align}
In the third and the fourth line $e_{\mu\nu}$ is localized at site $j=-1$.
Since $\omega$ is translation invariant, we have
\begin{align}\label{net}
(\ref{nagaii})
=\omega
\lmk
\bigotimes_{j=-l}^{-1}e_{\mu_j,\nu_j}\otimes e_{\mu\nu}
\otimes \bigotimes_{j=1}^{l}e_{\lambda_{j-1},\eta_{j-1}}
\rmk.
\end{align}
Here, $e_{\mu\nu}$ is localized at site $j=0$.
Then we have
\begin{align}
&(\ref{net})
=\braket{\Omega_\omega}{
\lmk
\pi_\omega\circ{\gamma_{L\to R}}\lmk\bigotimes_{j=-l}^{-1}e_{\mu_j,\nu_j}\rmk
\otimes \pi_\omega\lmk e_{\mu\nu}\otimes \bigotimes_{j=1}^{l}e_{\lambda_{j-1},\eta_{j-1}}\rmk
\rmk
\Omega_\omega}\nonumber\\
&=\braket{\Omega_\omega}{
\lmk
\pi_\omega\circ{\gamma_{L\to R}}\lmk\bigotimes_{j=-l}^{-1}e_{\mu_j,\nu_j}\rmk
\otimes 
\cs_\mu \cs_{\lambda_{0}}\cdots \cs_{\lambda_{l-1}}\cs_{\eta_{l-1}}^*\cdots \cs_{\eta_{0}}^*\cs_\nu^*
\rmk
\Omega_\omega}\nonumber\\
&=
\braket{\Omega_\omega}
{
\lmk
\pi_\omega\circ{\gamma_{L\to R}}\lmk\bigotimes_{j=-l}^{-1}e_{\mu_j,\nu_j}\rmk
\otimes 
\cs_\mu 
\pi_\omega\lmk \bigotimes_{j=0}^{l-1}e_{\lambda_j,\eta_j}\rmk
\cs_\nu^*
\rmk
{\Omega_\omega}}\nonumber\\
&=\braket{\lmk
\unit\otimes \cs_\mu^*
\rmk
\Omega_\omega}
{
\lmk
\pi_\omega\circ{\gamma_{L\to R}}\lmk\bigotimes_{j=-l}^{-1}e_{\mu_j,\nu_j}\rmk
\otimes 
\pi_\omega\lmk \bigotimes_{j=0}^{l-1}e_{\lambda_j,\eta_j}\rmk
\rmk
{\lmk \unit\otimes \cs_\nu^*\rmk\Omega_\omega}}
\end{align}
Hence we obtain
\begin{align}\label{nagai}
&\braket{\lmk \cs_\mu^*\otimes \unit\rmk\Omega_\omega}{
\lmk
\pi_\omega\circ{\gamma_{L\to R}}\lmk\bigotimes_{j=-l}^{-1}e_{\mu_j,\nu_j}\rmk
\otimes \pi_\omega\lmk\bigotimes_{j=0}^{l-1}e_{\lambda_j,\eta_j}
\rmk\rmk
\lmk \cs_\nu^* \otimes \unit\rmk
\Omega_\omega}\nonumber\\
&=
\braket{\lmk
\unit\otimes \cs_\mu^*
\rmk
\Omega_\omega}
{
\lmk
\pi_\omega\circ{\gamma_{L\to R}}\lmk\bigotimes_{j=-l}^{-1}e_{\mu_j,\nu_j}\rmk
\otimes 
\pi_\omega\lmk \bigotimes_{j=0}^{l-1}e_{\lambda_j,\eta_j}\rmk
\rmk
{\lmk \unit\otimes \cs_\nu^*\rmk\Omega_\omega}},
\end{align}
for any $l\in\nan$, $\mu,\nu, \mu_{-l},\ldots, \mu_{-1}, \nu_{-l},\ldots, \nu_{-1}, \lambda_0,\ldots,
\lambda_{l-1}, \eta_0,\ldots, \eta_{l-1}=1,\ldots,d$.
Since $\hat \pi_\omega(\caA_{\rm loc})$ is dense in $B(\hat \caH_\omega)$
with respect to the $\sigma$-weak topology,
this means
\begin{align}
\ketbra
{\lmk \cs_\nu^* \otimes \unit\rmk
\Omega_\omega}{\lmk \cs_\mu^*\otimes \unit\rmk\Omega_\omega}
=\ketbra{\lmk  \unit\otimes\cs_\nu^* \rmk
\Omega_\omega}
{\lmk
\unit\otimes \cs_\mu^*
\rmk
\Omega_\omega},
\end{align}
proving the claim.

From (\ref{ss4}) with  $\mu=\nu=1,\ldots,d$, we see that there is $a_\nu\in \bbT$
such that
\begin{align}
\lmk \cs_\nu^* \otimes \unit\rmk
\Omega_\omega
=a_\nu\lmk  \unit\otimes\cs_\nu^* \rmk
\Omega_\omega.
\end{align}
Substituting this to (\ref{ss4}), we find 
$a_\mu=a_\nu=:a$, if $\lmk \cs_\nu^* \otimes \unit\rmk
\Omega_\omega$, $\lmk \cs_\mu^* \otimes \unit\rmk
\Omega_\omega$ are not zero.
Hence we get a constant $a\in\bbT$ such that
\begin{align}
\lmk \cs_\nu^* \otimes \unit\rmk
\Omega_\omega
=a\lmk  \unit\otimes\cs_\nu^* \rmk
\Omega_\omega,\quad \nu=1,\ldots, d.
\end{align}
By the definition of $\bb$ and recalling $s_\omega=P_\caK$, we obtain 
\begin{align}\label{bv}
\lmk B_\nu^* \otimes \unit\rmk
\Omega_\omega
=a\lmk  \unit\otimes B_\nu^* \rmk
\Omega_\omega,\quad \nu=1,\ldots, d.
\end{align}

On the other hand, by (\ref{fdef}), (\ref{add}), (\ref{dac}), (\ref{jtheta}) we have
\begin{align}\label{floz}
&\lmk B_\nu^* \otimes \unit\rmk
\Omega_\omega
=\lmk B_\nu^* \otimes s_\omega\rmk
\Omega_\omega
=J_\omega\Delta_\omega^{-\frac 12} \lmk B_\nu\otimes s_\omega\rmk 
\Omega_\omega
=\Delta_{\omega}^{\frac 12}J_\omega
\lmk
B_\nu\otimes s_\omega 
\rmk J_\omega^*\Omega_\omega\nonumber\\
&=\Delta_{\omega}^{\frac 12}
\lmk
s_\omega\otimes\theta B_\nu\theta^*
\rmk\Omega_\omega
=\lmk
s_\omega\otimes \rho_\omega^{\frac 12} \theta  B_\nu\theta^* \rho_\omega^{-\frac 12}
\rmk\Omega_\omega.
\end{align}

Combining (\ref{bv}) and (\ref{floz}), we obtain
\begin{align}
\lmk
s_\omega\otimes \rho_\omega^{\frac 12} \theta  B_\nu\theta^* \rho_\omega^{-\frac 12}
\rmk\Omega_\omega
=a\lmk  s_\omega\otimes B_\nu^* \rmk
\Omega_\omega,\quad \nu=1,\ldots, d.
\end{align}
Since $\Omega_\omega$ is separating for $\caM$, we obtain
\begin{align}
\rho_\omega^{\frac 12} \theta  B_\nu\theta^* \rho_\omega^{-\frac 12}
=a B_\nu^*,\quad \nu=1,\ldots, d.
\end{align}
Taking adjoint, we obtain (\ref{bget}).
This completes the proof of the Theorem.
\end{proof}

{\bf Acknowledgment.}\\
{The author is grateful to Hal Tasaki for fruitful discussion which was essential for the present work, and for the helpful comments on the manuscript. The beginning of the introduction heavily relies on his help.
This work was supported by JSPS KAKENHI Grant Number 16K05171. 
}
\bigskip

%\begin{defn}\label{defind}
%For the rest of this paper,
%we will call $\sigma_\varphi$, the $\bbZ_2$-index associated to
%$\varphi$,
%the time reversal invariant pure split state on $\caA$.
%\end{defn}


\begin{thebibliography}{10}
\bibitem[AKLT]{Affleck:1988vr}
I.~Affleck, T.~Kennedy, E.H. Lieb, and H.~Tasaki.
\newblock {Valence bond ground states in isotropic quantum antiferromagnets}.
\newblock {\em Comm. Math. Phys.}, 115, 477--528, 1988.


\bibitem[A]{arv}W.B.~ Arveson. \newblock{Continuous analogues of Fock space I.}
\newblock{ Mem. Amer. Math. Soc.}, {\bf 409}, 1989.



\bibitem[BMNS]{bmns}
S.~ Bachmann, S.~Michalakis, B.~Nachtergaele, and R.~Sims.
\newblock{Automorphic Equivalence within Gapped Phases of Quantum Lattice Systems.}\newblock{Communications in Mathematical Physics}
{\bf 309}, 835--871, 2012. 



\bibitem[BN]{bn}
S.~ Bachmann and B.~Nachtergaele.
\newblock{On gapped phases with a continuous symmetry and boundary operators}
\newblock{J. Stat. Phy.}
{\bf 154}, 91--112, 2014. 

\bibitem[BJP]{bjp}
O.~Bratteli  P.~Jorgensen, G.~Price. 
\newblock{Endomorphisms of $B(\caH)$.}
\newblock{Quantization, nonlinear partial differential equations, and operator algebra.}
93--138, 
Proc. Sympos. Pure Math., {\bf 59}, 1996. 

\bibitem[BJ]{BJ}O.~Bratteli,
P. E. T.~Jorgensen.
\newblock{ Endomorphisms of $B(H)$
II. Finitely Correlated States on $O_n$.}
\newblock{Journal of functional analysis.} {\bf 145}, 323--373 1997.

\bibitem[BR1]{BR1}
 O.~Bratteli,  D. ~W.~Robinson.
\newblock {\em Operator Algebras and Quantum Statistical 
 Mechanics 1.} Springer-Verlag, 1986.
 \bibitem[BR2]{BR2}
 O.~Bratteli,  D.~W.~Robinson.
 \newblock {\em Operator Algebras and Quantum Statistical 
 Mechanics 2.} Springer-Verlag, 1996.
 
 \bibitem[DJP]{DJP}
J. Derezinski, V. Jaksic and C.-A. Pillet:
\newblock {Perturbation theory of $W^*$-dynamics, Liouvilleans and KMS-states}.
\newblock {\em Reviews in Mathematical Physics},15-05, 447--489, 2003.
 \bibitem[GW]{GuWen2009}
Z.-C.~Gu,  and X.-G.~Wen,
\newblock{\em Tensor-entanglement-filtering renormalization approach and symmetry-protected topological order},
Phys. Rev. B,  {\bf 80}, 155131 2009.
\bibitem[CGW]{ChenGuWEn2011}
X. Chen, Z.-C. Gu, and X.-G. Wen,
\newblock{ Classification of gapped symmetric phases in one-dimensional spin systems},
Phys. Rev. B {\bf 83}, 035107 2011.

\bibitem[DL]{dl}
S.~Doplicher, R.~Longo. \newblock{Standard and split inclusions of von Neumann algebras.} Invent. Math. 
{\bf 75} 493--536. 1984.
\bibitem[FNW]{Fannes:1992vq}
M.~Fannes, B.~Nachtergaele, and R.F. Werner.
\newblock {Finitely correlated states on quantum spin chains}.
\newblock {\em Comm. Math. Phys.}, {\bf 144}, 443--490, 1992.


\bibitem[FNW2]{fnwpure}
M.~Fannes, B.~Nachtergaele, and R.F. Werner.
\newblock {Finitely correlated pure states}.
\newblock {\em Journal of functional analysis.}, {\bf 120}, 511--534, 1994.

%\bibitem[H1]{h1}M.~Hastings.
%\newblock{Lieb-Schultz-Mattisinhigherdimensions}.
%\newblock{Phys.Rev.B}, 69,104431, 2004.


\bibitem[Hal1]{Haldane1983a}
F.D.M. Haldane, 
\newblock{Continuum dynamics of the 1-D Heisenberg antiferromagnet: identification with the $O(3)$ nonlinear sigma model},
Phys. Lett. {\bf 93A}, 464--468 1983.
%\\\url{http://www.sciencedirect.com/science/article/pii/037596018390631X}

\bibitem[Hal2]{Haldane1983b}
F.D.M. Haldane, 
\newblock{Nonlinear field theory of large-spin Heisenberg antiferromagnets: semiclassically quantized solitons of the one-dimensional easy-axis N\'eel state},
Phys. Rev. Lett. {\bf 50} 1153--1156 1983.
%\\\url{https://journals.aps.org/prl/abstract/10.1103/PhysRevLett.50.1153}

\bibitem[H1]{area}
M.~Hastings.
\newblock{An area law for one-dimensional quantum systems.}
\newblock{Journal of Statistical Mechanics.}  P08024, 2007.

\bibitem[H2]{h1}M.~Hastings.
\newblock{Quasi-adiabatic Continuation for Disordered Systems: Applications to Correlations, Lieb-Schultz-Mattis, and Hall Conductance.}
 http://arxiv/org/abs/1001.5280v2 [math-ph], 2010.

\bibitem[K]{Kennedy1990}
T. Kennedy,
\newblock{Exact diagonalization of open spin 1 chains},
J.~Phys.: Cond. Matt. {\bf 2}, 5737--5745, 1990.

\bibitem[KN]{kn}T.~Koma and B.~Nachtergaele
\newblock{The Spectral Gap of the Ferromagnetic XXZ-Chain}
\newblock{Letters in Mathematical Physics}
{\bf 40}, 1--16,
1997.

\bibitem[KT1]{kt} T.~Kennedy and H.~Tasaki.
\newblock{Hidden $\bbZ_2\times\bbZ_2$-symmetry breaking in Haldane-gap antiferromagnets.}
\newblock{Phys. Rev. B}, {\bf 45} 304--307, 1992.

 \bibitem[KT2]{kt2} T.~Kennedy and H.~Tasaki.
\newblock{Hidden symmetry breaking and the Haldane phase in $S= 1$ quantum spin chains.}
\newblock{Communications in Mathematical Physics}, {\bf 147} 431--484, 1992.


\bibitem[M1]{Matsui3}T.~Matsui.
\newblock{A characterization of matrix product pure states.}
\newblock{Infinite dimensional analysis and quantum probability.}
{\bf 1} 647--661. 1998.

\bibitem[M2]{Matsui1}
T.~Matsui.
\newblock{The split property and the symmetry breaking of the quantum spin chain.}
\newblock{Communications in Mathematical Physics}, {\bf 218}
393--416, 2001.


\bibitem[M3]{Matsui2}T.~Matsui.
\newblock{Boundedness of entanglement entropy and split property of quantum spin chains}.
\newblock{
Reviews in Mathematical Physics},
1350017, 2013.


\bibitem[NOS]{nos}B.~Nachtergaele, Y.~Ogata, and R.~Sims.
\newblock{Boundedness of entanglement entropy and split property of quantum spin chains}.
\newblock{J. Stat. Phys},
{\bf 124}, 1--13, 2006.

\bibitem[NR]{denNijsRommelse}
M. den Nijs and K. Rommelse,
\newblock{Preroughening transitions in crystal surfaces and valence-bond phases in quantum spin chains},
Phys. Rev. B {\bf 40}, 4709, 1989.


\bibitem[O1]{Ogata1}
Y.~Ogata.
\newblock {A class of asymmetric gapped Hamiltonians on quantum spin chains and its classification I}.
\newblock{Communications in Mathematical Physics}, {\bf 348}, 847--895, 2016.

\bibitem[O2]{Ogata2}
Y.~Ogata.
\newblock {A class of asymmetric gapped Hamiltonians on quantum spin chains and its classification II}.
\newblock{Communications in Mathematical Physics}, {\bf 348}, 897--957, 2016. 


\bibitem[O3]{Ogata3}
Y.~Ogata.
\newblock {A class of asymmetric gapped Hamiltonians on quantum spin chains and its classification III}.
\newblock{Communications in Mathematical Physics}, {\bf 352}, 1205--1263, 2017.

\bibitem[O4]{Ogata4}
Y.~Ogata.
\newblock{ A ${\mathbb Z}_2$-index of symmetry protected topological phases with
  time reversal symmetry for quantum spin chains}
\newblock{arXiv:1810.01045}
  
\bibitem[OT]{ot}Y.~Ogata and H.~Tasaki
\newblock{Lieb-Schultz-Mattis type theorems for quantum spin chains without continuous symmetry}.
\newblock{Communications in Mathematical Physics}, 2019.



\bibitem[PTBO1]{po}F.~Pollmann, A.~Turner, E.~Berg, and M.~Oshikawa
\newblock{Entanglement spectrum of a topological phase in one dimension}.
Phys. Rev. B {\bf 81}, 064439, 2010.

\bibitem[PTBO2]{po2}F.~Pollmann, A.~Turner, E.~Berg, and M.~Oshikawa
\newblock{Symmetry protection of topological phases in one-dimensional quantum spin systems}.
Phys. Rev. B {\bf 81}, 075125, 2012.

\bibitem[PWSVC]{Perez-Garcia2008}
D. Perez-Garcia, M.M. Wolf, M. Sanz, F. Verstraete, and J.I. Cirac,
{\em String order and symmetries in quantum spin lattices},\/
Phys. Rev. Lett. {\bf 100}, 167202 2008.

\bibitem[T1]{takesaki}
M.~Takesaki,
\newblock{Theory of operator algebras. I.}
\newblock{\em{Encyclopaedia of Mathematical Sciences}.} Springer-Verlag, Berlin, 2002.
\bibitem[T2]{takesaki2}
M.~Takesaki,
\newblock{Theory of operator algebras. II.}
\newblock{\em{Encyclopaedia of Mathematical Sciences}.} Springer-Verlag, Berlin, 2003.
%\bibitem[Tas1]{TasakiLSM}
%H. Tasaki,
%{\em Lieb-Schultz-Mattis theorem with a local twist for general one-dimensional quantum systems}\/,
%J. Stat. Phys. 170, 653--671 2018.

\bibitem[Tas1]{ta}
H.~Tasaki
\newblock{Topological phase transition and Z2 index for S = 1 quantum spin chains}
\newblock{arXiv:1804.04337}

\bibitem[Tas2]{TasakiBook}
H.~Tasaki,
{\em Physics and mathematics of quantum many-body systems}, (to be published from Springer).


\bibitem[W]{Wolf:2012aa}
M.M. Wolf.
\newblock{Quantum channels {\&} operations.}
\newblock{ Unpublished.}  2012.



\end{thebibliography}
\end{document}